\documentclass[12pt,notitlepage]{article}

\usepackage{enumitem}% http://ctan.org/pkg/enumerate
\usepackage[margin=1in]{geometry}
\usepackage{relsize}
\usepackage{graphicx}
\usepackage[colorlinks=true,citecolor=black,linkcolor=black,urlcolor=blue]{hyperref}
\hypersetup{colorlinks={true},linkcolor={blue},citecolor=green}
\usepackage{mathtools}
\usepackage{mathpazo}
\usepackage{mathrsfs}
\usepackage{url}
\usepackage{amsthm,amsmath,amssymb}
\usepackage{enumitem}
\usepackage[nameinlink]{cleveref}
\usepackage{authblk}

\theoremstyle{plain}
\newtheorem{theorem}{Theorem}
\newtheorem{lemma}[theorem]{Lemma}
\newtheorem{corollary}[theorem]{Corollary}
\newtheorem{proposition}[theorem]{Proposition}

\newtheorem{observation}[theorem]{Observation}

\theoremstyle{definition}
\newtheorem{definition}[theorem]{Definition}
\newtheorem{example}[theorem]{Example}

\newtheorem{problem}[theorem]{Problem}
\newtheorem{question}[theorem]{Question}
\newtheorem{remark}[theorem]{Remark}

\theoremstyle{remark}

\crefname{problem}{problem}{problems}
\crefname{proposition}{proposition}{propositions}
\crefname{question}{question}{questions}

\newlist{probenum}{enumerate}{1} % also creates a counter called 'probenumi'
\setlist[probenum]{label=\alph*), ref=\theproposition~(\alph*)}
\crefalias{probenumi}{problem} 

\newlist{thmenum}{enumerate}{1} % also creates a counter called 'probenumi'
\setlist[thmenum]{label=\alph*), ref=\theproposition~(\alph*)}
\crefalias{thmenumi}{theorem} 

\newlist{propenum}{enumerate}{1} % also creates a counter called 'probenumi'
\setlist[propenum]{label=\alph*), ref=\theproposition~(\alph*)}
\crefalias{propenumi}{proposition} 
\title{Waring Rank, Parameterized and Exact Algorithms}
\author{Kevin Pratt \footnote{\href{mailto:kpratt@andrew.cmu.edu}{kpratt@andrew.cmu.edu}}}
\affil{Computer Science Department, Carnegie Mellon University}
\begin{document}
\maketitle
\begin{abstract}
Given nonnegative integers $n$ and $d$, where $n \gg d$, what is the minimum number $r$ such that there exist linear forms $\ell_1, \ldots, \ell_r \in \mathbb{C}[x_1, \ldots, x_n]$ so that $\ell_1^d + \cdots + \ell_r^d$ is supported exactly on the set of all degree-$d$ multilinear monomials in $x_1, \ldots, x_n$? We show that this and related questions have surprising and intimate connections to the areas of parameterized and exact algorithms, generalizing several well-known methods and providing a concrete approach to obtain faster approximate counting and deterministic decision algorithms. This gives a new application of Waring rank, a classical topic in algebraic geometry with connections to algebraic complexity theory, to computer science.

To illustrate the amenability and utility of this approach, we give a randomized $4.075^d \cdot \mathrm{poly}(n, \varepsilon^{-1})$-time algorithm for computing a $(1 + \varepsilon)$ approximation of the sum of the coefficients of the multilinear monomials in a degree-$d$ homogeneous $n$-variate polynomial with nonnegative coefficients. As an application of this we give a faster algorithm for approximately counting subgraphs of bounded treewidth, improving on earlier work of Alon et al. Along the way we give an exact answer to an open problem of Koutis and Williams and sharpen a lower bound on the size of perfectly balanced hash families given by Alon and Gutner.
\end{abstract}
\clearpage

\section{Introduction}
The \textit{Waring rank} of a homogeneous $n$-variate degree-$d$ polynomial $f \in \mathcal{S}_d^n \coloneqq \mathbb{C}[x_1, \ldots, x_n]_d$, denoted $\mathbf{R}(f)$, is the minimum $r$ such that
\begin{equation}\label{wdecomp}
f = \ell_1^d + \cdots + \ell_r^d,
\end{equation}
for some linear forms $\ell_1 , \ldots , \ell_r \in \mathcal{S}_1^n$. The study of Waring rank is a classical problem in algebraic geometry and invariant theory, with pioneering work done in the second half of the 19th century by A. Clebsch, J.J. Sylvester, and T. Reye, among others \cite[Introduction]{iarrobino1999power}. It has enjoyed a recent resurgence of popularity within algebraic geometry \cite{iarrobino1999power,landsberg2012tensors} and has connections in computer science to the limiting exponent of matrix multiplication $\omega$ \cite{chiantini2018polynomials}, the Mulmuley-Sohoni Geometric Complexity Theory program \cite{burgisser2019no}, and several other areas in algebraic complexity \cite{landsberg_2017,efremenko2018barriers}. This paper adds \textit{parameterized algorithms} to this list, showing that several methods in this area (color-coding methods \cite{alon1995color,alon2007balanced,huffner2008algorithm}, the group-algebra/determinant sum approach \cite{koutis2008faster,williams2009finding,bjorklund2010determinant}, and inclusion-exclusion methods) fundamentally result from rank upper bounds for a specific family of polynomials. In a situation analogous to that of $\omega$, better explicit upper bounds on the Waring rank of these polynomials yield faster algorithms for certain problems in a completely black-box manner, and lower bounds on the Waring rank of these polynomials imply barriers such algorithms face.%Among other things, this provides the first concrete approach that we are aware of for derandomizing and improving several parameterized decision and approximate counting algorithms.

This connection should not come as a complete surprise, as many algorithms work by solving a question about the coefficients of some efficiently-computable ``generating polynomial'' determined by the input. The insight of this paper, which has been largely unexploited, is that in general this is a question about Waring rank.

Let $e_{n,d} \coloneqq \sum_{1 \le i_1 < i_2 < \cdots < i_d \le n} x_{i_1} \cdots x_{i_d}$ denote the elementary symmetric polynomial of degree $d$ in $n$ variables. We will study the following questions:

\begin{question}\label{supprquest}
What is $A(n,d)$, the minimum Waring rank among all $g \in \mathcal{S}_d^n$ with the property that $\mathrm{supp}(g) = \mathrm{supp}(e_{n,d})$?\footnote{Here $\mathrm{supp}( \sum_{\alpha \in \mathbb{N}^n} c_\alpha x_1^{\alpha_1} \cdots x_n ^{\alpha_n}) \coloneqq \{\alpha \in \mathbb{N}^n : c_\alpha \neq 0 \}$.}
\end{question}

\begin{question}\label{apndq}
What is $A^+(n,d)$, the the minimum Waring rank among all $g \in \mathbb{R}_{\ge 0}[x_1, \ldots, x_n]$ with the property that $\mathrm{supp}(g) = \mathrm{supp}(e_{n,d})$?
\end{question}

\begin{question}\label{aendq}
For $0 \le \varepsilon < 1$, what is $A^\epsilon(n,d)$, the minimum Waring rank among all $g~\in~\mathbb{R}[x_1, \ldots, x_n]$ with the property that $\mathrm{supp}(g) = \mathrm{supp}(e_{n,d})$ and the nonzero coefficients of $g$ are in the range $1 \pm \varepsilon$?
\end{question}

We now illustrate the algorithmic relevance of these questions with a new and very simple $\binom{n}{\lfloor d/2 \rfloor}\mathrm{poly}(n)$-time and $\mathrm{poly}(n)$-space algorithm for exactly counting simple cycles (i.e., closed walks with no repeated vertices) of length $d$ in an $n$-vertex graph. This is the fastest polynomial space algorithm for this problem, improving on a $2^d \binom{n}{\lfloor d/2 \rfloor}\mathrm{poly}(n)$-time algorithm of Fomin et al.~\cite{fomin2012faster} which in turn improved on a $2^d (d/2)! \binom{n}{\lfloor d/2 \rfloor}\mathrm{poly}(n)$-time algorithm of Vassilevska Williams and Williams \cite{vassilevska2009finding}.

Given a directed graph $G$, let $A_G$ be the symbolic matrix with entry $(i,j)$ equal to the variable $x_i$ if there is an edge from vertex $v_i$ to vertex $v_j$, and zero otherwise. By the trace method,
\begin{equation}\label{gencycles}
f_G \coloneqq \mathrm{tr}(A_G^d) = \sum_{\substack{\mathrm{closed~walks} \\ (v_{i_1},v_{i_2},\ldots,v_{i_d}) \in G}} x_{i_1} \cdots x_{i_d} \in \mathcal{S}_d^n.
\end{equation}
Now we denote by $g(\partial \mathbf{x})$ the partial differential operator $g(\frac{\partial}{\partial x_1}, \ldots, \frac{\partial}{\partial x_n})$. The algorithm is based on two easy observations:

\begin{observation}
The number of simple cycles of length $d$ in $G$ equals $e_{n,d}(\partial \mathbf{x})f_G$.
\end{observation}

\begin{observation}\label{psrev}
If $g = a_1 \ell_1^d + \cdots + a_d \ell_r^d$, where $\ell_i = c_{i,1} x_1 + \cdots + c_{i,n}x_n$ for $i = 1, \ldots , r$, then for all $f \in \mathcal{S}_d^n$,
\[g(\partial \mathbf{x})f = d! \sum_{i=1}^r a_i f(c_{i,1} , \ldots , c_{i,n}).\]
\end{observation}

It is immediate that we can compute the number of simple cycles in $G$ of length $d$ using $\mathbf{R}(e_{n,d}) = A^0(n,d)$ evaluations of $f_G$. Now, it was recently shown in \cite{lee2016power} that 
\[\mathbf{R}(e_{n,d}) \le \binom{n}{\le \lfloor d/2 \rfloor}\coloneqq \sum_{i=0}^{\lfloor d/2 \rfloor} \binom{n}{i}.\]
Explicitly, for $S \subseteq [n]$ and $i \in [n]$, define the indicator function $\delta_{S,i} \coloneqq -1$ if $i \in S$, and $\delta_{S,i} \coloneqq 1$ otherwise. Then for $d$ odd,
\[2^{d-1}d!\cdot e_{n,d} = \sum_{\substack{S \subset [n] \\ |S| \le \lfloor d/2 \rfloor}} (-1)^{|S|} \binom{n-\lfloor d/2 \rfloor-|S|-1}{\lfloor d/2 \rfloor-|S|}(\delta_{S,1} x_1 + \delta_{S,2} x_2 + \cdots + \delta_{S,n} x_n)^d.\]
(A similar formula holds for $d$ even.) It follows that the number of length-$d$ simple cycles in $G$ equals
\begin{equation}\label{lolpath}\frac{1}{2^{d-1}}\sum_{\substack{S \subset [n] \\ |S| \le \lfloor d/2 \rfloor}} (-1)^{|S|}\binom{n-\lfloor d/2 \rfloor-|S|-1}{\lfloor d/2 \rfloor-|S|} f_G(\delta_{S,1}, \ldots , \delta_{S,n}).
\end{equation}
This gives a closed form for the number of length-$d$ simple cycles in $G$ that is easily seen to be computable in the stated time and space bounds. This algorithm is much simpler, both computationally and conceptually, than those of previous approaches.\footnote{We note that the use of inclusion-exclusion  (or ``M\"{o}bius inversion'' \cite{nederlof2009fast}) in numerous exact-counting algorithms, such as Ryser's formula for computing the permanent \cite{leech_1964} and algorithms for counting Hamiltonian cycles \cite{kohn1977generating} and set packings \cite{bjorklund2006inclusion}, implicitly relies on a natural but suboptimal bound on $\mathbf{R}(e_{n,d})$; namely the one given by \Cref{ryser} below. We elaborate on this in \Cref{ryserex}.}

%\footnote{In contemporaneous work, \cite{recent} have given an alternate solution to this problem.}

The above argument shows something very general: given $f \in \mathcal{S}_d^n$ as a black-box, we can compute $e_{n,d}(\partial \mathbf{x})f$ (that is, the sum of the coefficients of the \emph{multilinear monomials} in $f$)  using $\binom{n}{\le \lfloor d/2 \rfloor}$ queries. This answers a ``significant'' open problem asked by Koutis and Williams \cite{koutis2009limits} in a completely black-box way.\footnote{An alternate solution to this problem was given contemporaneously in \cite{recent}.} Moreover, it follows from a special case of our \Cref{opteval} that \emph{any} algorithm must make $\mathbf{R}(e_{n,d}) \ge \Omega(\binom{n}{\le \lfloor d/2 \rfloor})$ \cite{lee2016power} queries to compute $e_{n,d}(\partial \mathbf{x})f$ in the black-box setting:
\begin{theorem}\label{opteval}
Fix $g \in \mathcal{S}_d^n$ and let $f \in \mathcal{S}_d^n$ be given as a black-box. The minimum number of queries to $f$ needed to compute $g(\partial \mathbf{x}) f$ is $\mathbf{R}(g)$, assuming unit-cost arithmetic operations.
\end{theorem}
In light of this lower bound, one might next ask for a $(1 \pm \varepsilon$) approximation of $e_{n,d}(\partial \mathbf{x})f$. This prompts our main algorithmic result, which is based on an answer to \Cref{aendq}:

\begin{theorem}\label{appl1}
Let $f \in \mathbb{R}_{\ge 0}[x_1, \ldots, x_n]_d$ be given as a black-box. There is a randomized algorithm which given any $0 < \varepsilon < 1$ computes a number $z$ such that with probability 2/3, 
\[(1- \varepsilon) \cdot e_{n,d}(\partial \mathbf{x})f < z < (1+\varepsilon)\cdot e_{n,d}(\partial \mathbf{x})f.\]
This algorithm runs in time $4.075^d \cdot \varepsilon^{-2}\log(\varepsilon^{-1}) \cdot \mathrm{poly}(n,s_f)$ and uses $\mathrm{poly}(n,s_f,\log (\varepsilon^{-1}))$ space. Here $s_f$ is the maximum bit complexity of $f$ on the domain $\{\pm 1\}^n$.
\end{theorem}

%Both the algorithm and the proof behind \Cref{appl1} are very simple and can be found in \Cref{sec4}.
The algorithm and the proof behind \Cref{appl1} are simple and can be found in \Cref{sec4}. Applying this theorem to to the graph polynomial $f_G$, an algorithm for approximately counting simple cycles of length $d$ is immediate.\footnote{In fact, \Cref{appl1} gives the fastest \emph{polynomial space} algorithm for approximately counting cycles that we are aware of.} More generally, we have the following:

\begin{theorem}\label{twcount}
Let $G$ and $H$ be graphs where $|G| = n$, $|H| = d$, and $H$ has treewidth $\mathrm{tw}(H)$. There is a randomized algorithm which given any $0<\varepsilon <1$ computes a number $z$ such that with probability $2/3$, 
\[(1-\varepsilon) \cdot \mathrm{Sub}(H,G) < z < (1+\varepsilon) \cdot \mathrm{Sub}(H,G).\]
This algorithm runs in time $4.075^d \cdot n^{\mathrm{tw}(H)+O(1)} \cdot \varepsilon^{-2}\log (\varepsilon^{-1})$. Here $\mathrm{Sub}(H,G)$ denotes the number of subgraphs of $G$ isomorphic to $H$.
\end{theorem}

This is the fastest known algorithm for approximately counting subgraphs of bounded treewidth, improving on a $5.44^d  n^{\mathrm{tw}(H)+O(1)} \varepsilon^{-2}$-time algorithm of Alon et al. \cite{alon2008biomolecular} which in turn improved on a $5.44^{d\log \log d} n^{\mathrm{tw}(H)+O(1)} \varepsilon^{-2}$-time algorithm of Alon and Gutner \cite{alon2007balanced}. The first parameterized algorithm for a variant of this problem was given by Arvind and Raman \cite{arvind2002approximation} and has runtime $d^{O(d)}n^{\mathrm{tw}(H)+O(1)}$. In the case that $H$ has pathwidth $\mathrm{pw}(H)$, a recent algorithm of Brand et al.~\cite{Brand2018Extensorcoding} runs in time $4^d n^{\mathrm{pw}(H)+O(1)} \varepsilon^{-2}$. We stress that this application is only a motivating example -- \Cref{appl1} is extremely general and also applies to counting set partitions and packings \cite{bjorklund2006inclusion}, dominating sets \cite{koutis2009limits}, repetition-free longest common subsequences \cite{blin2012parameterized}, and functional motifs in biological networks \cite{guillemot2013finding}. 

 In the rest of this section we outline our approach. This will suggest a path to derandomize and improve the base of the exponent in \Cref{appl1} (and hence \Cref{twcount}) from $4.075$ to $2$. Specifically, we raise the following question:

\begin{question}
Is $A^\varepsilon(n,d) \le 2^d \cdot \mathrm{poly}(n,\varepsilon^{-1})$?
\end{question}

Prior to this work it was believed \cite{koutis2015algebraic} that a derandomization of polynomial identity testing would be needed to obtain, for instance, a deterministic $2^d \mathrm{poly}(n)$-time algorithm just for \emph{detecting} simple paths of length $d$ in a graph. On the contrary, an explicit affirmative answer to the above question would give a $2^d \mathrm{poly}(n, \varepsilon^{-1})$-time deterministic algorithm for \emph{approximately counting} simple paths.

\begin{remark}\label{gurvcase}
A focus on approximating $g(\partial \mathbf{x})f$ in the case that $f$ and $g$ are real stable has recently led to several advances in algorithms and combinatorics; see e.g. \cite{gurvits2006hyperbolic}. In particular, a result of Anari et al.~\cite{anari2016nash} shows that in this case $e_{n,d}(\partial \mathbf{x})f$ can be approximated (up to a factor of $e^{d + \varepsilon}$) deterministically in polynomial time given black-box access to $f$. This paper shows that the general (i.e., \emph{unstable}) case raises interesting questions as well.
\end{remark}

\subsection{Our Approach and Connections to Previous Work}
To continue with the previous example, note that the graph polynomial $f_G$ is supported on a multilinear monomial if and only if $G$ contains a cycle of length $d$. This motivates the following problem of well-recognized algorithmic importance \cite{gurvits2004combinatorial,koutis2008faster,williams2009finding}:

\begin{problem}\label{detectmultilin}
Given black-box access to $f \in \mathcal{S}_d^n$ over $\mathbb{C}^n$, decide if $f$ is supported on a multilinear monomial.
\end{problem}

%The relevance of this problem to parameterized algorithms was first made explicit by Koutis \cite{koutis2005faster,koutis2008faster}. Koutis's approach was then used by Williams \cite{williams2009finding} to give a randomized $2^d\mathrm{poly}(n)$-time algorithm for detecting simple paths of length $d$ in an $n$-vertex directed graph.\footnote{Henceforth called the $d$-path problem. We go against the trend of calling this the $k$-path problem to emphasize that the parameter corresponds to the degree of a polynomial.} This improved on a long line of research \cite{monien1985find,kneis2006divide,alon1995color,jia2004efficient}, and remains the fastest known algorithm for this problem. This approach is now responsible for a handful of the fastest parameterized and exact algorithms \cite{koutis2015algebraic}, in particular a $1.66^d \cdot \mathrm{poly}(n)$-time algorithm for the $d$-path problem in the case of undirected graphs \cite{bjorklund2010narrow}. We note that the first parameterized algorithm for the general white-box version of \Cref{detectmultilin} was only given recently by Brand et al. \cite{Brand2018Extensorcoding}, and the case of \Cref{detectmultilin} when $n = d$ and $f \in \mathbb{R}_{\ge 0}[x_1, \ldots, x_n]$ is real-stable was studied earlier in a separate context \cite{gurvits2004combinatorial} (see also \Cref{gurvcase}).

It is not hard to see that any algorithm for computing $g(\partial \mathbf{x})f$, where $g$ is supported on exactly the set of degree-$d$ multilinear monomials, can be used to solve \Cref{detectmultilin} with one-sided error (\Cref{supptesta}). This suggests studying upper bounds on $A(n,d)$ (\Cref{supprquest}) as an approach to solve \Cref{detectmultilin}. Perhaps surprisingly though, it turns out that several known methods in parameterized algorithms can be understood as giving constructive upper bounds on $A(n,d)$, and better upper bounds to $A(n,d)$ would improve upon these methods. For example, the seminal color-coding method of Alon, Yuster, and Zwick \cite{alon1995color} can be recovered from an upper bound on $A(n,d)$ of $O(5.44^d \log n)$, and an improvement to color-coding given by H\"{u}ffner et al.~\cite{huffner2008algorithm} follows from an upper bound on $A(n,d)$ of $ O(4.32^d \log n)$ (\Cref{colorred}). The group-algebra/determinant sum approach of \cite{williams2009finding,koutis2008faster,bjorklund2010determinant} reduces to answering a generalization of \Cref{supprquest} (see \Cref{genand}) in the case that the underlying field is not $\mathbb{C}$ but of characteristic 2. (In \Cref{char2ez} we give the essentially optimal upper bound of $2^d - 1$ for this variant, which in turn can be used to recover \cite{williams2009finding,koutis2008faster,bjorklund2010determinant}). Prior to this work, no connection of this precision between these methods was known.

\Cref{supprquest} provides insight into lower bounds on previous methods as well. For example, the bounds on $\mathbf{R}(e_{n,d})$ given in \cite{lee2016power} directly yield asymptotically sharper lower bounds than those given by Alon and Gutner \cite[Theorem 1]{alon2009balanced} on the size of \emph{perfectly balanced hash families} used by exact-counting color-coding algorithms (\Cref{hashlower}). Curiously, this improvement is ultimately a consequence of \emph{B\'ezout's theorem} in algebraic geometry. \Cref{supprquest} and a classical lower bound on Waring rank (\Cref{methodderivs}) explain why \emph{disjointness matrices} arose in the context of lower bounds on color-coding \cite{alon2009balanced} and the group-algebra approach \cite{koutis2009limits}: they are the partial derivatives matrices of the elementary symmetric polynomials.

Our main answers to \Cref{supprquest} are the following. By our \Cref{lowerbound,catupper,pracub}, it follows that 
\[2^{d-1} \le A(n,d) \le \min(6.75^d, O(4.075^d \log n)).\]
Perhaps surprisingly, this gives an upper bound on $A(n,d)$ \emph{independent} of $n$. On the negative side, our lower bound on $A(n,d)$ rules out \Cref{supprquest} as an approach to obtain algorithms faster than $2^d \mathrm{poly}(n)$ for \Cref{detectmultilin}; moreover, we show in \Cref{multlintest} that there is also a lower bound of $2^{d-1}$ on the number of queries needed to solve \Cref{detectmultilin} with one-sided error. 

It is easily seen by \Cref{psrev} that constructive upper bounds on $A^+(n,d)$ yield deterministic algorithms for determining if $f$ is supported on a multilinear monomial in the case that $f$ has nonnegative real coefficients (as, e.g., the graph polynomial $f_G$ has), and constructive upper bounds on $A^\varepsilon(n,d)$ yield deterministic algorithms for approximating $e_{n,d}(\partial \mathbf{x}) f$. This broadly generalizes the use of color-coding in designing approximate counting and deterministic decision algorithms. 

Our bounds on $A(n,d)$ also hold for $A^+(n,d)$. Remarkably, we show in \Cref{d410} that if $A^+(33700,4) \le 10$ then $A^+(n,d) \le O(3.9999^d \log n)$. It follows from our \Cref{lowerbound} and \Cref{pracub} that 
\[2^{d-1} \le A^\varepsilon(n,d) \le O(4.075^d\varepsilon^{-2} \log n ),\]
and from our \Cref{aepsiloninf} that $\lim_{n \to \infty}A^\varepsilon(n,d) = \infty$ for all $d > 1$ and $\varepsilon<1/2$ -- unlike $A^+(n,d)$, $A^\varepsilon(n,d)$ depends on $n$. As an aside, it is immediate that \[\underline{\mathbf{R}}(e_{n,d}) \le \lim_{\varepsilon \to 0} A^\varepsilon(n,d) \le \mathbf{R}(e_{n,d}),\]
where $\underline{\mathbf{R}}(g)$ denotes the \emph{Waring border rank} of $g$, i.e., the minimum $r$ such that there exists a sequence of polynomials of Waring rank at most $r$ converging to $g$ in the Euclidean (or equivalently, Zariski) topology. 

\subsection{Paper Overview}
For ease of exposition, we work over $\mathbb{C}$ unless specified otherwise. Most of our theorems can be extended to infinite (or sufficiently large) fields of arbitrary characteristic by replacing the polynomial ring with the ring of divided power polynomials (see \cite[Appendix A]{iarrobino1999power}). Except for in \Cref{sec4}, we assume that arithmetic operations can be performed with infinite precision and at unit cost.

In \Cref{sec2} we introduce concepts related to Waring rank (in particular the \emph{Apolarity Lemma}) in order to better understand the following problems:

\begin{problem}%\label{probinterp}
Fix $g\in \mathcal{S}_d^n$. Given black-box access to $f \in \mathcal{S}_d^n$,
\begin{probenum}%[label=\roman*.]
\item Compute $g(\partial \mathbf{x})f$.\label{p1a}
\item Compute a $(1 \pm \varepsilon)$ approximation of $g(\partial \mathbf{x})f$ (assuming $f,g \in \mathbb{R}_{\ge 0}[x_1, \ldots, x_n]$).\label{p1b}
\item Determine if $\mathrm{supp}(f) \cap \mathrm{supp}(g) = \emptyset$.\label{p1c}
\end{probenum}
\end{problem}
The fundamental connection between Waring rank and \Cref{p1a} is given by our \Cref{opteval}. Using similar ideas, we show that at least $2^{d-1}$ queries are required to test if $\mathrm{supp}(f) \cap \mathrm{supp}(e_{n,d}) = \emptyset$ with one-sided error in \Cref{multlintest}. We then introduce the new concepts of support rank, $\varepsilon$-support rank, and nonnegative support rank, which give upper bounds on the complexity of randomized and deterministic algorithms for \Cref{p1a,p1b,p1c}. A related notion of support rank for tensors has previously appeared in the context of $\omega$ and quantum communication complexity \cite{cohn2013fast,blaser2017border,walter2016multi}, but we are unaware of previous work on support rank in the symmetric (polynomial) case. In the case when $d=2$ these notions are related to the well-studied concepts of sign rank, zero-nonzero rank, and approximate rank of matrices \cite{barak2011rank, alon2013approximate}.

In \Cref{sec3} we study $A(n,d)$ and its variants. We start in \Cref{sec3.1} by proving negative results, showing that $A(n,d) \ge 2^{d-1}$ (\Cref{lowerbound}),  and that for sufficiently large $n$, $A(n,2) = 3$ (\Cref{d2case}) and $A(n,3) \ge 5$ (\Cref{d3case}). Using bounds on the \emph{$\varepsilon$-rank} of the identity matrix \cite{alon2003problems}, we show in \Cref{aepsiloninf} that for $1/\sqrt{n}\le \varepsilon < 1/2$, 
\[\Omega(\log n \cdot \varepsilon^{-2}/\log(\varepsilon^{-1})) \le A^\varepsilon(n,2) \le O(\log n \cdot \varepsilon^{-2}).\]
While it may at first seem like we are splitting hairs by focusing on particular values of $d$, we will later show in \Cref{d410} that, for example, proving that $A^+(n,4) \le 10$ for sufficiently large $n$ would yield improved upper bounds on $A^+(n,d)$ for \emph{all} $n$ and $d$. 

Curiously, our lower bound on $A(n,3)$ is a consequence of the classical \emph{Cayley-Salmon theorem} in algebraic geometry, and our general lower bound on $A(n,d)$ ultimately follows from B\'ezout's theorem via \cite{ranestad2011rank}. On this note, we show in \Cref{fermat} that \Cref{supprquest} is equivalent to a question about the geometry of linear spaces contained in the \emph{Fermat hypersurface} $\{x \in \mathbb{C}^n : \sum_{i=1}^n x_i^d = 0\}$.

The rest of \Cref{sec3} is focused on general upper bounds on $A(n,d)$ and its variants. \Cref{detgen} will give a simple explanation as to why \emph{determinant sums} (as in the title of \cite{bjorklund2010determinant}) can be computed in a parameterized way: for all $d \times n$ matrices $A$ and $B$, the Waring rank of
\begin{equation}\label{detsumpoly}
\sum_{\substack{\alpha \in \{0,1\}^n \\ |\alpha| = d}} \det(A_\alpha B_\alpha) x_1^{\alpha_1} \cdots x_n^{\alpha_n}
\end{equation}
is at most $\mathbf{R}(\det_d)$. A special case of this example is used in \Cref{catupper} to show that $A^+(n,d) < 6.75^d$. In order to improve this, it would suffice to find a better upper bound on the Waring rank of a single polynomial: the determinant of a symbolic $d \times d$ Hankel matrix. We show in \Cref{catlower} that the method of partial derivatives cannot give lower bounds on the Waring rank of this polynomial better than $2.6^d$. 

Next we define rank for polynomials over a field $\mathsf{k}$ of arbitrary characteristic -- as it is, our definition of rank is not valid in positive characteristic (example: try to write $xy$ as a sum of squares of linear forms over a field of characteristic two). Using this we define $A_\mathsf{k}(n,d)$, which equals $A(n,d)$ when $\mathrm{char}(\mathsf{k}) = 0$. We note in \Cref{generallower} that $A_\mathsf{k}(n,d) \ge 2^{d-1}$. \Cref{char2ez} shows that this lower bound is essentially optimal when $\mathrm{char}(\mathsf{k}) = 2$, as then $A_\mathsf{k}(n,d) \le 2^d-1$; specifically, this rank upper bound holds for \Cref{detsumpoly} in the case that $A = B$. This is a simple consequence of the fact that the permanent and the determinant agree in characteristic 2. We explain in this section how the group-algebra approach of \cite{koutis2008faster,williams2009finding} and the basis of \cite{bjorklund2010determinant} reduce to a slightly weaker fact than this upper bound. A precise connection between support rank and a certain ``product-property'' of abelian group algebras critical to \cite{koutis2008faster,williams2009finding} is given by \Cref{groupalg}. 

In \Cref{sec3.3} we present a method for translating upper bounds on $A^+(n_0,d_0)$ for some \emph{fixed} $n_0$ and $d_0$ into upper bounds on $A^+(n,d)$ for \emph{all} $n$ and $d$ (\Cref{masterthm}). This method also allows us to recursively bound $A^\varepsilon(n,d)$ for fixed $d$ (\Cref{aepsilonbound}). This approach can be seen as a vast generalization of color-coding methods, and is based on a \emph{direct power sum} operation on polynomials and a combinatorial tool generalizing \emph{splitters} that we call a \emph{perfect splitter}. We use this to show that $A^\varepsilon(n,d) \le O(4.075^d \varepsilon^{-2} \log n)$ in \Cref{pracub}.

In \Cref{sec4} we give applications of the previous section. We start by giving the proof \Cref{appl1}, which is then used to prove \Cref{twcount}. We end with an improved lower bound on the size of perfectly-balanced hash families in \Cref{hashlower}.

We conclude by giving several standalone problems.
\section{Preliminaries and Methods}\label{sec2}

We use multi-index notation: for $f \in \mathcal{S}_d^n$, we write $f  =  \sum_{\alpha \in \mathbb{N}^n} c_\alpha x^\alpha$, where $x^\alpha \coloneqq x_1^{\alpha_1} \cdots x_n^{\alpha_n}$.  For $\alpha \in \mathbb{N}^n$, we let $|\alpha| \coloneqq \sum_{i=1}^n \alpha_i$ and  $\alpha! \coloneqq \alpha_1! \alpha_2! \cdots \alpha_n!$. We then define $\mathbb{N}_d^n \coloneqq \{\alpha \in \mathbb{N}^n : |\alpha| = d\}$, and similarly $\{0,1\}^n_d \coloneqq \{\alpha \in \{0,1\}^n : |\alpha| = d\}$. Given $\beta \in \mathbb{N}^n$ we say that $\alpha \ge \beta$ if $\alpha_i \ge \beta_i$ for all $i \in [n]$. We denote by $\partial_i$ the differential operator $\frac{\partial}{\partial x_i}$, and we let $\partial^\alpha \coloneqq \partial_1^{\alpha_1} \cdots \partial_n^{\alpha_n}$. We let $\mathbf{V}(f) \coloneqq \{p \in \mathbb{C}^n : f(p) = 0\}$ denote the hypersurface defined by $f$. For $\ell = \sum_{i=1}^n a_i x_i \in \mathcal{S}_1^n$, we let $\ell^* \coloneqq (a_1, \ldots, a_n) \in \mathbb{C}^n$. For $X \subseteq \mathbb{C}^n$, the ideal of polynomials in $\mathcal{S}^n$ vanishing on $X$ is denoted by $I(X)$. The ideal generated by $f_1, \ldots, f_k \in \mathcal{S}^n$ is denoted by $\langle f_1, \ldots , f_k \rangle$. Given an ideal $I \subseteq \mathcal{S}^n$ we let $I_d$ denote the subspace of $I$ of degree-$d$ polynomials.

The set of $n \times m$ matrices with entries in a field $\mathsf{k}$ is denoted by $\mathsf{k}^{n \times m}$. For a matrix $A \in \mathsf{k}^{n \times m}$ and a multi-index $\alpha \in \mathbb{N}^n$, we let $A_\alpha$ be the $n \times |\alpha|$ matrix whose first $\alpha_1$ columns are the first column of $A$, next $\alpha_2$ columns are the second column of $A$, etc. We let $\det_d$, $\mathrm{per}_d \in \mathsf{k}[x_{ij} : i,j \in [d]]_d$ denote the degree-$d$ determinant and permanent polynomials, respectively. Recall that the permanent is defined by
\[\mathrm{per}_d = \sum_{\sigma \in \mathfrak{S}_d} \prod_{i=1}^d x_{i,\sigma(i)},\]
where $\mathfrak{S}_d$ denotes the symmetric group on $d$ letters.

The subsequent theorems are classical and easily verified. The first is the crux of this paper. The second shows that Waring rank is always defined (i.e., finite).
\begin{theorem}
\label{gendual}
Let $f \in \mathcal{S}_d^n$ and let $j \ge d$.
\begin{thmenum}
\item \cite[Lemma 1.15(i)]{iarrobino1999power} For all $\ell_1, \ldots, \ell_r \in \mathcal{S}^n_1$,  \label{genduala}
\[f(\partial \mathbf{x}) \sum_{i=1}^r \ell_i^j = d!\sum_{i=1}^r f(\ell_i^*)\ell_i^{j-d}.\]
\item \cite[Lemma 3.5]{comon2008symmetric} For all $g \in \mathcal{S}_d^n$, $f(\partial \mathbf{x}) g = g(\partial \mathbf{x}) f$.\label{gendualb}
\end{thmenum}
\end{theorem}

\begin{theorem}\cite[Corollary 1.16]{iarrobino1999power}
$\mathbf{R}(f) \le \dim \mathcal{S}_d^n = \binom{n+d-1}{d}$.\label{gendualc}
\end{theorem}

Importantly, \Cref{genduala,gendualb} imply that $g(\partial \mathbf{x})f$ can be computed with $\mathbf{R}(g)$ queries in \Cref{p1a}, as noted in \Cref{psrev}.
We will show in the next subsection that this is optimal, even if we are allowed to query $f$ adaptively. 

\begin{example}\label{ryserex}
The following Waring decomposition of $e_{n,d}$ is easily seen by inclusion-exclusion:
\begin{equation}\label{ryser}
d! \cdot e_{n,d} = \sum_{\substack{\alpha \in \{0,1\}^n \\ |\alpha| \le d}}(-1)^{|\alpha|+d} \binom{n-|\alpha|}{d-|\alpha|} \left  (\sum_{i=1}^n \alpha_i x_i\right )^d.
\end{equation}
In fact, this decomposition is \emph{synonymous} with inclusion-exclusion in many exact algorithms, as we now illustrate. For $A \in \mathbb{C}^{n \times n}$, let 
\[\mathrm{Prod}_A \coloneqq (A_{1,1}x_1 + \cdots + A_{1,n}x_n) \cdots (A_{n,1}x_1 + \cdots + A_{n,n} x_n) \in \mathcal{S}_n^n.\]
 It is easily seen that the coefficient of $x_1 \cdots x_n$ in $\mathrm{Prod}_A$ equals the permanent of $A$. In other words, $\mathrm{per}(A) = e_{n,n}(\partial \mathbf{x}) \mathrm{Prod}_A$. It follows directly from \Cref{gendual} and \Cref{ryser} that
\[\mathrm{per}(A) = \sum_{\alpha \in \{0,1\}^n} (-1)^{|\alpha|+n} \mathrm{Prod}_A(\alpha),\]
which is Ryser's formula for computing the permanent \cite{leech_1964}. As another example, applying \Cref{gendual} and \Cref{ryser} to the closed-walk generating polynomial \Cref{gencycles}, one finds that the number of Hamiltonian cycles in $G$ equals
\[\sum_{\alpha \in \{0,1\}^n} (-1)^{|\alpha|+n} \mathrm{tr}(A_G^n)(\alpha),\]
which was first given in \cite{kohn1977generating} and rediscovered several times thereafter \cite{karp1982dynamic,bax1993inclusion}. As a third example, let $S_1, \ldots, S_m \subseteq [k \cdot r]$, where $|S_i| = r$ for all $i$. Note that that the coefficient of $x_1 \cdots x_{kr}$ in $\mathrm{Part}_{S_1, \ldots, S_m} \coloneqq \left (\sum_{i=1}^m \prod_{j \in S_i} x_j \right )^k$ equals the number of ordered partitions of $[kr]$ into $k$ of the sets $S_i$. Therefore the number of such partitions equals
\[\sum_{\alpha \in \{0,1\}^{kr}}(-1)^{|\alpha|+kr} \mathrm{Part}_{S_1, \ldots, S_m}(\alpha),\]
which was given in \cite{bjorklund2006inclusion,bjorklund2009set}. The fastest known algorithms for computing the permanent and counting Hamiltonian cycles and set partitions follow from the straightforward evaluation of the above formulas. A similar perspective on these algorithms appeared earlier in \cite{barvinok1996two}.

Understanding these algorithms from the perspective of Waring decompositions is extremely insightful, and was our initial motivation. For example, it is clear from the above argument that \emph{any} Waring decomposition of $x_1 \cdots x_n$ yields an algorithm for the above problems -- there is nothing special about \Cref{ryser}. This immediately raises the question: what is $\mathbf{R}(x_1 \cdots x_n)$? This was only answered recently in \cite{ranestad2011rank}, where a lower bound on the \emph{degree} of a form's \emph{apolar subscheme} was used to show that $\mathbf{R}(x_1 \cdots x_n) = 2^{n-1}$.\footnote{A lower bound of $\binom{n}{\lfloor n/2 \rfloor}$ can be shown easily using the \emph{method of partial derivatives}, presented in the next subsection.} This lower bound shows that the above algorithms are, in a restricted sense, optimal. Similar observations have been made in \cite{gurvits,glynn2013permanent}.

Although the Waring decomposition of \Cref{ryser} is essentially optimal in the case when $n = d$, it is far from optimal in general. Indeed, \Cref{ryser} only shows that $\mathbf{R}(e_{n,d}) \le \binom{n}{\le d}$, whereas it was shown in \cite{lee2016power} that for $d$ odd, $\mathbf{R}(e_{n,d}) = \binom{n}{\le \lfloor d/2 \rfloor}$,
and for $d$ even,
\[ \binom{n}{\le d/2}- \binom{n-1}{d/2} \le \mathbf{R}(e_{n,d}) \le \binom{n}{\le d/2}.\]
%The algorithms that result from using the upper bound given in \cite{lee2016power} are significantly less natural from a \emph{combinatorial} perspective than those resulting from \Cref{ryser}. We believe that this is the reason why they were not discovered earlier.
\end{example}

\subsection{Apolarity and the Method of Partial Derivatives}
Fix $g \in \mathcal{S}_d^n$. For integers $u,v \ge 0$ such that $u+v=d$, let $\mathit{Cat}_g(u,v) : \mathcal{S}_u^n \to \mathcal{S}_v^n$ be given by 
\[\mathit{Cat}_g(u,v)(f) \coloneqq  f(\partial \mathbf{x})g.\]
These maps, called \textit{catalecticants}, were first introduced by J.J. Sylvester in 1852 \cite{sylvester1970principles}. Their importance is due in large part to the following method for obtaining Waring rank lower bounds, known as the \textit{method of partial derivatives} in complexity theory \cite[Section 6.2.2]{landsberg_2017}.

\begin{theorem}\cite[pg. 11]{iarrobino1999power} \label{methodderivs}
For all $g \in \mathcal{S}_d^n$ and integers $u,v \ge 0$ such that $u+v=d$, 
\[\mathbf{R}(g) \ge \mathrm{rank}(\mathit{Cat}_g(u,v)).\]
\end{theorem}
%\begin{proof}
%Suppose that $g = \sum_{i=1}^r \ell_i^d$. Then for any $f \in \mathcal{S}_u^n$, $f(\partial \mathbf{x}) g = u! \sum_{i=1}^r f(\ell_i^*) \ell_i^v$  by \Cref{genduala}. This shows that the image of $\mathit{Cat}_g(u,v)$ is contained in the span of $\ell_1^{v}, \ldots, \ell_r^{v}$, and therefore $\mathrm{rank}(\mathit{Cat}_g(u,v)) \le r \le \mathbf{R}(g)$.
%\end{proof}

\begin{remark}\label{derivsopt}
As a matrix, $\mathit{Cat}_g(u,v)$ has $\binom{n+u-1}{u}$ columns, indexed by the degree-$u$ monomials in $x_1, \ldots, x_n$, and $\binom{n+v-1}{v}$ rows, indexed by the degree-$v$ monomials in $x_1, \ldots, x_n$. Therefore the best rank lower bound \Cref{methodderivs} can give is $\binom{n+\lceil d/2 \rceil -1}{\lceil d/2 \rceil}$, which is obtained when $u = \lceil d/2 \rceil, v = \lfloor d/2 \rfloor$. In contrast, it is known \cite[Section 3.2]{landsberg2012tensors} that the rank for \textit{almost all} $g \in \mathcal{S}_d^n$  is at least $\lceil \binom{n+d-1}{d}/n \rceil$ (with respect to a natural distribution on forms), so the method of partial derivatives is far from optimal. Finding methods for proving better lower bounds is a significant barrier and a topic of great interest from both an algebraic-geometric and complexity-theoretic perspective; see \cite[Section 10.1]{landsberg_2017} and \cite{efremenko2018barriers}.
\end{remark}

\begin{example}\label{matrixrk}
It is a classical fact from linear algebra that for $g \in \mathcal{S}_2^n$, $\mathbf{R}(g) = \mathrm{rank}(\mathit{Cat}_g(1,1))$. Explicitly, this says that $g = \sum_{1 \le i \le j \le n} A_{ij} x_i x_j$ can be written as a sum of at most $r$ squares of linear forms if and only if the matrix $A = (A_{ij})$ has rank at most $r$. Hence Waring rank can be viewed as a higher dimensional generalization of symmetric matrix rank.
\end{example}

Let $g^\perp_j \coloneqq \ker \mathit{Cat}_g(j,d-j)$ be the set of degree-$j$ forms annihilating $g$ under the differentiation action. The next fact is known as the \emph{Apolarity Lemma} in the Waring rank literature.

\begin{lemma}\cite[Theorem 4.2]{teitler2014geometric}
Let $\ell_1, \ldots, \ell_r \in \mathcal{S}_1^n$ be pairwise linearly independent. Then for all $g \in \mathcal{S}_d^n$, $g \in \mathrm{span}\{\ell_1^d, \ldots, \ell_r^d\}$ if and only if $I(\{\ell_1^*, \ldots, \ell_r^*\})_d \subset g^\perp_d$.
\end{lemma}

A complete answer to the complexity of \Cref{p1a} is now in hand.
\begin{proof}[Proof of \Cref{opteval}]
The upper bound is immediate from \Cref{gendualb}. To prove the lower bound we first show the following: for any pairwise linearly independent points $v_1, \ldots, v_m \in \mathbb{C}^n$ where $m < \mathbf{R}(g)$, there exists a $p \in \mathcal{S}_d^n$ such that $p \in I(\{v_1, \ldots, v_m\})$ but $g(\partial \mathbf{x}) p \neq 0$. If this were not the case, there exist pairwise linearly independent points $v_1, \ldots, v_m$ such that $I(\{v_1, \ldots, v_m\})_d \subset g^\perp_d$. But this implies that $g$ has rank at most $m$ by the Apolarity Lemma, a contradiction.

So now given any $f \in \mathcal{S}_d^n$, suppose that our algorithm queries $f$ at $v_1, \ldots, v_m$, which can be assumed to be pairwise linearly independent. By the above argument, there exists some $p \in \mathcal{S}_d^n$ such that $(p+f)(v_i) = p(v_i) + f(v_i) = f(v_i)$ for all $i \in [m]$, and hence the algorithm cannot distinguish $f$ from $p+f$, but at the same time $g(\partial \mathbf{x})f \neq g(\partial \mathbf{x})(p+f) $.
\end{proof}

\subsection{Support Rank, Nonnegative Support Rank, and $\varepsilon$-Support Rank}
We now introduce variants of Waring rank of algorithmic relevance.

\begin{definition}\label{supprdef}
The support rank and nonnegative support rank of $f \in \mathcal{S}_d^n$ are given by
\begin{align*}
\mathbf{R}_{\mathrm{supp}}(f) &\coloneqq \min (\mathbf{R}(g) : g \in \mathcal{S}_d^n, \mathrm{supp}(g) = \mathrm{supp}(f)),\\
\mathbf{R}_{\mathrm{supp}}^+(f) &\coloneqq  \min (\mathbf{R}(g) : g \in \mathbb{R}_{\ge 0}[x_1, \ldots, x_n]_d,\mathrm{supp}(g) = \mathrm{supp}(f)).\\
\intertext{Furthermore, if $f \in \mathbb{R}_{\ge 0}[x_1, \ldots, x_n]_d$, the $\varepsilon$-support rank of $f$ is given by}
\mathbf{R}_{\mathrm{supp}}^\varepsilon(f) &\coloneqq  \min (\mathbf{R}(g) : g \in \mathbb{R}[x_1, \ldots, x_n]_d, \forall \alpha \in \mathbb{N}^n_d, (1-\varepsilon) \cdot \partial^\alpha f \le \partial^\alpha g \le (1+\varepsilon) \cdot \partial^\alpha f).
\end{align*}
\end{definition}
Note that condition in the definition of $\mathbf{R}_{\mathrm{supp}}^\varepsilon$ is simply that the coefficient of $x^\alpha$ in $g$ is bounded by a factor of $(1 \pm \varepsilon)$ times the coefficient of $x^\alpha$ in $f$.

Roughly speaking, support rank corresponds to decision algorithms, nonnegative support rank to \emph{deterministic} decision algorithms, and $\varepsilon$-support rank to deterministic \emph{approximate counting} algorithms. This is now formalized.

\begin{definition}
For $g \in \mathcal{S}_d^n$ and $0 < \delta < 1$, a $g$-\textit{support intersection certification algorithm} with one-sided error $\delta$ is an algorithm which, given any $f \in \mathcal{S}_d^n$ as a black-box, outputs $``\mathrm{supp}(f) \cap \mathrm{supp}(g) = \emptyset''$ on all instances $f$ where $\mathrm{supp}(f) \cap \mathrm{supp}(g) = \emptyset$, and correctly outputs $``\mathrm{supp}(f) \cap \mathrm{supp}(g) \neq \emptyset''$ with probability at least $1-\delta$ on all instances where $\mathrm{supp}(f) \cap \mathrm{supp}(g) \neq \emptyset$.
\end{definition}

\begin{proposition}\label{supptest}
\begin{propenum}
\item For all $g \in \mathcal{S}_d^n$ and $\delta > 0$, there is a $g$-support intersection certification algorithm with one-sided error $\delta$ that makes $\mathbf{R}_{\mathrm{supp}}(g)$ queries.\label{supptesta}
\item For a fixed $g \in \mathcal{S}_d^n$ and all $f \in \mathbb{R}_{\ge 0}[x_1, \ldots, x_n]_d$ given as a black-box, there is a deterministic algorithm that decides if $\mathrm{supp}(g) \cap \mathrm{supp}(f)$ using $\mathbf{R}_{\mathrm{supp}}^+(g)$ queries.\label{supptestb}
\item For a fixed $g \in \mathbb{R}_{\ge 0}[x_1, \ldots, x_n]_d$ and all $f \in \mathbb{R}_{\ge 0}[x_1, \ldots, x_n]_d$ given as a black-box, there is a deterministic algorithm that computes a $(1 \pm \varepsilon)$-approximation to $g(\partial \mathbf{x})f$ using $\mathbf{R}_{\mathrm{supp}}^\varepsilon(g)$ queries.\label{supptestc}
\end{propenum}
\end{proposition}
\begin{proof}
\begin{enumerate}[label=\alph*.]
\item Let $U \subseteq \mathbb{C}$, where $|U| \ge d/\delta$. Let $a_1, \ldots, a_n$ be indeterminates. Note that $g(\partial \mathbf{x}) f(a_1 x_1, \ldots,a_n x_n)$ is not identically zero in $\mathbb{C}[a_1, \ldots, a_n]$ if and only if $\mathrm{supp}(f) \cap \mathrm{supp}(g) \neq \emptyset$. Then by choosing $a_1, \ldots, a_n$ uniformly at random from $U$, $g(\partial \mathbf{x})$ $f(a_1 x_1, \ldots,a_n x_n)$ will evaluate to zero whenever $\mathrm{supp}(f) \cap \mathrm{supp}(g) = \emptyset$, and whenever $\mathrm{supp}(f) \cap \mathrm{supp}(g) \neq \emptyset$ this does not evaluate to zero with probability at least $1-\delta$ by the Schwartz-Zippel lemma. By \Cref{gendual}, $g(\partial \mathbf{x}) f(a_1 x_1, \ldots,a_n x_n)$  can be computed using $\mathbf{R}(g)$ queries, and the conclusion follows.
\item If both $f$ and $g$ have nonnegative coefficients, then $g(\partial \mathbf{x})f > 0$ if and only if $\mathrm{supp}(f) \cap \mathrm{supp}(g) \neq \emptyset$. The result follows from \Cref{gendual}.
\item This is immediate from \Cref{gendual}.\qedhere
\end{enumerate}
\end{proof}

It follows from a variation of the proof of \Cref{opteval} that \Cref{supptesta} is optimal for monomials:

\begin{proposition}\label{montest}
For all $\alpha \in \mathbb{N}^n$ and all $\delta < 1$, any $x^\alpha$-support intersection certification algorithm with one-sided error $\delta$ makes at least $\mathbf{R}_{\mathrm{supp}}(x^\alpha) = \prod_{i=1}^n (1+\alpha_i)/\min_{i \in [n]}(1+\alpha_i)$ queries.
\end{proposition}
\begin{proof}
The upper bound follows from \Cref{gendualb}; in fact, this shows that we can compute $\partial^\alpha f$ exactly using $\mathbf{R}(x^\alpha)$ queries.

For the lower bound, given any $f \in \mathcal{S}_d^n$ where $\alpha \in \mathrm{supp}(f)$, suppose a support intersection certification algorithm queries $f$ at pairwise linearly independent points $v_1, \ldots, v_m$, where $m < \mathbf{R}(x^\alpha)$. Then by the Apolarity Lemma, there exists a $p \in \mathcal{S}_d^n$ such that $p \in I(\{v_1, \ldots, v_m\})$ but $\partial^\alpha p \neq 0$ (see the proof of \Cref{opteval}). Note that the condition that $\partial^\alpha p \neq 0$ is equivalent to saying that $\alpha \in \mathrm{supp}(p)$. Therefore there exists some $\lambda \in \mathbb{C}$ such that $\alpha \notin \mathrm{supp}(f + \lambda p)$. But note that $(f+\lambda p)(v_i) = f(v_i) + \lambda p(v_i) = f(v_i)$ for all $i \in [m]$, and hence the algorithm cannot distinguish between $f$ and $f + \lambda p$. Since the algorithm has no false negatives, it must always give the incorrect answer on $f$. We conclude by the matching upper and lower bounds on $\mathbf{R}(x^\alpha)$ given in \cite{carlini2011solution}.
\end{proof}

\begin{theorem}\label{multlintest}
Any $e_{n,d}$-support intersection certification algorithm with one-sided error $\delta$ makes at least $2^{d-1}$ queries.
\end{theorem}
\begin{proof}
Suppose for contradiction that such an algorithm made fewer queries. Then given $f$ as a black-box, we run this algorithm with access to $f(x_1, \ldots, x_d, 0, \ldots, 0)$. By definition, this algorithm always answers correctly if the coefficient of $x_1 \cdots x_d$ is zero, and answers correctly with probability at least $1-\delta$ if this coefficient is nonzero. But this gives an $x_1 \cdots x_d$-support intersection certification algorithm with one-sided error $\delta$ making fewer than $2^{d-1}$ queries. Since $\mathbf{R}(x_1 \cdots x_d) = 2^{d-1}$ \cite{ranestad2011rank}, this contradicts \Cref{{montest}}.
\end{proof}
\section{Support Ranks of Elementary Symmetric Polynomials}\label{sec3}
We are now ready to study $A(n,d)$ and its variants, which we now recall. 
\begin{problem}\label{eltsrank}
Determine $A(n,d) \coloneqq \mathbf{R}_{\mathrm{supp}}(e_{n,d})$, $A^+(n,d) \coloneqq \mathbf{R}^+_{\mathrm{supp}}(e_{n,d})$ and $A^\varepsilon(n,d) \coloneqq \mathbf{R}_{\mathrm{supp}}^\varepsilon(e_{n,d})$.
\end{problem}

Obviously $A(n,d) \le A^+(n,d) \le A^\varepsilon(n,d),$ and for all $n$, $A(n,1) = 1$. It follows from \cite{ranestad2011rank} that $A^\varepsilon(n,n) = 2^{n-1}$ and from \cite{lee2016power} that $A^\varepsilon(n,d) \le \binom{n}{\le \lfloor d/2 \rfloor}$; the latter turns out to be arbitrarily far from optimal, however.

We will be interested in \Cref{eltsrank} as $n$ goes to infinity. To facilitate this, we adopt the notation $A(\mathbb{N},d) \coloneqq \lim_{n \to \infty} A(n,d)$, defining $A^+(\mathbb{N},d)$ and $A^\varepsilon(\mathbb{N},d)$ analogously. We will show in \Cref{trivialboundsa} that $A(n,d), A^+(n,d)$, and $A^\varepsilon(n,d)$ are nondecreasing in $n$, in \Cref{detgen} that $A^+(\mathbb{N},d)$ is finite for each $d$, and in \Cref{aepsiloninf} that $A^\varepsilon(\mathbb{N},d)$ is infinite for $\varepsilon<1/2$ and $d>1$.

For notational convenience, we define
\begin{align*}
\mathfrak{E}(n,d) &\coloneqq \{f \in \mathcal{S}_d^n : \mathrm{supp}(f) = \mathrm{supp}(e_{n,d})\},\\
\mathfrak{E}^+(n,d) &\coloneqq \{f \in \mathfrak{E}(n,d) : \forall \alpha \in \{0,1\}^n_d,\ \partial^\alpha f \in \mathbb{R}^+\},\\
\mathfrak{E}^\varepsilon(n,d) &\coloneqq \{f \in \mathfrak{E}^+(n,d) : \forall \alpha \in \{0,1\}^n_d,\ \partial^\alpha f \in (1 \pm \varepsilon))\}.
 \end{align*}
 \begin{remark}
Our upper bounds to \Cref{eltsrank} will be obtained by the following general method. We start with some $f \in\mathcal{S}_d^m$ whose rank is known. We then find $L_1, \ldots,L_m \in \mathcal{S}_1^n$, where $n \gg m$, so that $f(L_1, \ldots, L_m) \in \mathfrak{E}(n,d)$. This will show that 
\[A(n,d) \le \mathbf{R}(f(L_1, \ldots, L_m)) \le \mathbf{R}(f).\]
For example, we first show that $A^+(\mathbb{N},d)< 6.75^d$ by taking $f$ to be the determinant of a generic Hankel matrix, and $\ell_1, \ldots , \ell_n$ to be given by rank-1 Hankel matrices (points on the \textit{rational normal scroll}). We later use this method to show that $A^\varepsilon(n,d) \le O(4.075^d \varepsilon^{-2} \log n)$ by taking $f$ to be a ``direct sum'' of $e_{\lfloor 1.55d \rfloor,d}$ and $L_1, \ldots, L_n$ to be given by a $(1+\varepsilon)$-balanced splitter. We note in \Cref{colorred} that color-coding can be viewed as taking $f$ to be a direct sum of $x_1 x_2 \cdots x_d$ and $L_1, \ldots, L_m$ to be a perfect hash family. A simple geometric property that $f$ and $L_1, \ldots, L_m$ must satisfy in this method is given by \Cref{geomchar}.
\end{remark}
\subsection{Lower Bounding $A(n,d)$ and the $d=2$ Case}\label{sec3.1}
We start with some simple relations between different values of $A(n,d)$ that will be used throughout this section.
\begin{proposition}\label{trivialbounds}
For all $n \ge d$, 
\begin{propenum}
\item $A(n,d) \le A(n+1,d)$,\label{trivialboundsa}
\item $A(n,d) \le A(n+1,d+1)$.\label{trivialboundsb}
\end{propenum}
Moreover, these statements remain valid when ``$A$'' is replaced with $A^+$ and $A^\varepsilon$.
\end{proposition}
\begin{proof}
\begin{enumerate}[label=\alph*.]
\item Suppose $f \in \mathfrak{E}(n+1,d)$, and let $f'$ be obtained from $f$ by setting $x_{n+1} =0$. Then clearly $\mathbf{R}(f') \le \mathbf{R}(f)$ and $f' \in \mathfrak{E}(n,d)$. Therefore $A(n,d) \le A(n+1,d)$.
\item If $f \in \mathfrak{E}(n+1,d+1)$, then $\partial_{n+1} f \in \mathfrak{E}(n,d)$. Hence $A(n,d) \le \mathbf{R}(\partial_{n+1} f) \le \mathbf{R}(f)$, where the final inequality follows from \Cref{genduala}.
\end{enumerate}
It is easy to see that the same arguments hold if we replace $A(n,d)$ with $A^+(n,d)$ or $A^\varepsilon(n,d)$.
\end{proof}

\begin{theorem}\label{lowerbound}
For all $n \ge d$,
\[2^{d-1} \le A(n,d) \le A^+(n,d) \le A^\varepsilon(n,d).\]
\end{theorem}
\begin{proof}
It was shown in \cite{ranestad2011rank} that $\mathbf{R}(x_1 \cdots x_d) = 2^{d-1}$, and therefore $A(d,d) = 2^{d-1}$. The theorem is then immediate from \Cref{trivialboundsa}.
\end{proof}
We now give an insightful geometric characterization of $A(n,d)$.

\begin{proposition}\label{geomchar}
$A(n,d) \le r$ if and only if for some $m$ there exists $f \in \mathcal{S}_d^m$ and points $v_1, \ldots, v_n$ in $\mathbb{C}^m$ such that $\mathbf{R}(f) \le r$ and $f$ vanishes on the span of any $d-1$ of the points $v_1, \ldots, v_n$, but not on the span of any $d$ of them.
\end{proposition}
\begin{proof}
Suppose that $A(n,d) \le r$. By definition, there exists a $f \in \mathfrak{E}(n,d)$ with $\mathbf{R}(f) \le r$. It follows that $f$ vanishes on the span of the span of any $d-1$ of the standard basis vectors in $\mathbb{C}^n$, but not on the span of any $d$ of them.

Conversely, suppose there exists such an $f$ and points $v_1, \ldots, v_n$, and let 
\[f' \coloneqq f(x_1v_1 + \cdots + x_n v_n).\]
It is immediate that $\mathbf{R}(f') \le \mathbf{R}(f)$. Additionally, $f'$ must be multilinear as $f$ vanishes on the span of any $d-1$ of the points $v_1, \ldots, v_n$. But then for $\alpha \in \{0,1\}^n_d$, the coefficient of $x^\alpha$ in $f'$ is given by $f'(\alpha) = f(\sum_{i=1}^n \alpha_i v_i)$. If this was zero $f$ would vanish on the span of the $d$ points $\{v_i : i \in \mathrm{supp}(\alpha)\}$, a contradiction. This shows that $f' \in \mathfrak{E}(n,d)$, proving the claim.
\end{proof}

\begin{proposition}\label{fermat}
$A(n,d) \le r$ if and only if there exist $n$ points in $\mathbb{C}^r$ such that the span of any $d-1$ of them is contained in $\mathbf{V}(\sum_{i=1}^r x_i^d)$, but the span of any $d$ of them is not.
\end{proposition}

\begin{proof}
If $A(n,d) \le r$, then for some $f \in \mathfrak{E}(n,d)$ and linear forms $\ell_1, \ldots, \ell_r$, $f = \sum_{i=1}^r \ell_i^d$. Let $v_j \coloneqq (({\ell_1}^*)_j, ({\ell_2}^*)_j, \ldots, ({\ell_r}^*)_j)$ for all $j \in [n]$. Since $f$ is multilinear, $\sum_{i=1}^r x_i^d$ must vanish on the span of any $d-1$ of the points $v_1, \ldots, v_n$, and since each multilinear monomial has a nonzero coefficient, $\sum_{i=1}^r x_i^d$ does not vanish on the span of any $d$ of  $v_1, \ldots, v_n$. 

Conversely, suppose that there exists such a set of points. Since $\sum_{i=1}^r x_i^d$ has rank $r$, by \Cref{geomchar} we conclude that $A(n,d) \le r$. \qedhere
\end{proof}

\begin{corollary}\label{d3case}
$5 \le A(8,3) \le A(\mathbb{N},3)$.
\end{corollary}
\begin{proof}
Suppose for contradiction that $A(8,3) =  4$. By \Cref{fermat}, this implies that there are 8 points in $\mathbb{C}^4$ such that the planes spanned by any two of them are contained in $\mathbf{V}(x_1^3 + x_2^3 + x_3^3 + x_4^3)$, but the span of any three of them is not. Note that this is only possible if no three points are coplanar, and hence the $\binom{8}{2} = 28$ planes spanned by any two points are distinct. But by the Cayley-Salmon theorem, $\mathbf{V}(x_1^3 + x_2^3 + x_3^3 + x_4^3)$ contains exactly $27<28$ lines in the projective space $\mathbb{CP}^3$ \cite[Lemma 11.1]{gathmann2002algebraic}, a contradiction.
\end{proof}

\begin{remark}
A similar proof fails to show that $6 \le A(\mathbb{N},3)$, as $\mathbb{P}(\mathbf{V}(x_1^3 + \cdots + x_5^3))$ contains infinitely many lines (see \cite[Exercise 11.10.b]{gathmann2002algebraic}).
\end{remark}

The $d=2$ case of \Cref{eltsrank} is solved using linear algebra.

\begin{proposition}\label{d2case}
$A(\mathbb{N},2) = 3$.
\end{proposition}
\begin{proof}
It suffices by \Cref{matrixrk} to show that for $n \ge 3$, the minimum rank of a symmetric $n \times n$ matrix with zeros on the diagonal and nonzero values elsewhere is 3. There is a lower bound of 3 since the principal $3 \times 3$ minor of any such matrix is easily seen to be nonzero. An upper bound of 3 is given by the matrix $((i-j)^2)_{i,j \in [n]}$.
\end{proof}

To understand $A^\varepsilon(n,2)$ we will need the following fact:
\begin{theorem}\cite[Theorem 9.3]{alon2003problems}\label{thanksalon}
Let $B$ be an $n$-by-$n$ real matrix with $b_{i,i} = 1$ for all $i$ and $|b_{i,j}| \le \varepsilon$ for all $i \neq j$. Then if $1/\sqrt{n} \le \varepsilon < 1/2$,
\[\mathrm{rank}(B) \ge \Omega\left (\frac{\log n \cdot \varepsilon^{-2}}{\log(\varepsilon^{-1})} \right ).\]
\end{theorem}

\begin{proposition}\label{aepsilon2}
\begin{propenum}
\item If $1/\sqrt{n} \le \varepsilon < 1/2$, \label{aepsilon2a}
\[A^\varepsilon(n,2) \ge \Omega\left (\frac{\log n \cdot \varepsilon^{-2}}{\log(\varepsilon^{-1})} \right ).\]
\item For all $\varepsilon>0$, 
\[A^\varepsilon(n,2) \le O\left (\log n \cdot \varepsilon^{-2} \right ).\]\label{aepsilon2b}
\end{propenum}
\end{proposition}

\begin{proof}
It follows from \Cref{matrixrk} that $A^\varepsilon(n,2)$ is the minimum rank among all real symmetric matrices $A$ with $A_{i,i} = 0$ and $A_{i,j} \in [1-\varepsilon,1+\varepsilon]$ for all $i \neq j$. Note that given any such $A$, the matrix $J - A$ (where $J$ denotes the all-ones matrix) has diagonal entries equal to 1, off-diagonal entries bounded in absolute value by $\varepsilon$, and rank at most rank$(A)+1$. Conversely, given any symmetric matrix $B$ with $b_{i,i} = 1$ for all $i$ and $|b_{i,j}| \le \varepsilon$ for all $i \neq j$, the matrix $J-B$ has zeros on the diagonal, off-diagonal entries in the range $[1-\varepsilon,1+\varepsilon]$, and rank at most $\mathrm{rank}(B)+1$. So it suffices to determine the minimum rank of such a matrix $B$. Given this observation, (a) is immediate from \Cref{thanksalon}. 

To show (b), let $m \coloneqq O(\log n/\varepsilon^2)$. By the Johnson-Lindenstrauss Lemma, there exist unit vectors $v_1, \ldots, v_n \in \mathbb{R}^m$ such that $|v_i \cdot v_j| \le \varepsilon$ for all $i \neq j$. It follows that the matrix $(v_i^T \cdot v_j)_{i,j \in [n]}$ has the desired properties and rank at most $m$.
\end{proof}
\begin{corollary}\label{aepsiloninf}
For all $0 < \varepsilon < 1/2$ and $d \ge 2$, $A^\varepsilon(\mathbb{N},d) = \infty$.
\end{corollary}
\begin{proof}
Fix $0 < \varepsilon < 1/2$. By \Cref{aepsilon2a}, $A^\varepsilon(n,2)  \ge \Omega\left (\frac{\log n \cdot \varepsilon^{-2}}{\log(\varepsilon^{-1})} \right )$ for all $n \ge \varepsilon^{-2}$, and so $A^\varepsilon(\mathbb{N},2) = \infty$. Now suppose that $A^\varepsilon(\mathbb{N},d)$ is bounded above for some $d > 2$. Then by \Cref{trivialbounds}, for all $n$
\[A^\varepsilon(n,2) \le A^\varepsilon(n+d-2,d) \le A^\varepsilon(\mathbb{N},d),\]
a contradiction.
\end{proof}
\subsection{Upper Bounds via the Determinant}\label{sec3.2}
The relevance of the determinant to \Cref{eltsrank} is immediate from \Cref{geomchar}. The obvious but key observation is that for all $n, d$ with $n \ge d$, a generic set of $n$ rank-1 $d \times d$ matrices has the property that the sum of any $d$ of them is invertible, and hence the span of any $d-1$ of them is contained in $\mathbf{V}(\det_d)$ but the span of any $d$ of them is not. Applying \Cref{geomchar}, we conclude that $A(n,d) \le \mathbf{R}(\det_d)$. We now make this more explicit.
\begin{definition}
Let $d \le n$. For $A,B \in \mathbb{C}^{d \times n}$, let 
\begin{equation}\label{detsum}
g_{A,B} \coloneqq \sum_{\alpha \in \{0,1\}^n_d} \text{det}_d(A_\alpha B_\alpha)  x^\alpha.
\end{equation}
\end{definition}

\begin{proposition}\label{detgen}
For all $A ,B \in \mathbb{C}^{d \times n}$, 
\[\mathbf{R}(g_{A,B}) \le \mathbf{R}(\mathrm{det}_d) \le (5/6)^{\lfloor d/3 \rfloor} 2^{d-1} d!.\]
Furthermore, $A^+(\mathbb{N},d) \le \mathbf{R}(\mathrm{det}_d) \le (5/6)^{\lfloor d/3 \rfloor} 2^{d-1} d!$ and $A^+(\mathbb{N},d)$ exists.
\end{proposition}

\begin{proof}
Let $X = \mathrm{diag}(x_1, \ldots, x_n)$. By the Cauchy-Binet formula it follows that $\det_d((A\cdot~X) \cdot B^T)~=~g_{A,B}$. The first statement then follows from the fact that $\mathbf{R}(\mathrm{det}_d) \le  (5/6)^{\lfloor d/3 \rfloor} 2^{d-1} d!$ \cite[Example 1.14]{teitler2014geometric}.

Note that by taking $A$ and $B$ to have positive minors\footnote{For instance, by taking the columns of $A$ and $B$ to be given by real Vandermonde vectors.}, $g_{A,B} \in \mathfrak{E}^+(n,d)$. This shows that $A^+(\mathbb{N},d) \le \mathbf{R}(\det_d)$. Since \Cref{trivialboundsa} shows that $(A^+(n,d))_n$ is nondecreasing, it follows that the limit $A^+(\mathbb{N},d)$ exists.
\end{proof}

\begin{remark}
The asymptotically best known lower bound on $\mathbf{R}(\det_d)$ is $\binom{d}{\lfloor d/2 \rfloor}^2$, which follows from the method of partial derivatives \cite{gurvits} \cite[Theorem 9.3.2.1]{landsberg2012tensors}. Therefore one cannot hope to improve the upper bound given by \Cref{detgen} exponentially beyond $4^d$ by finding a better upper bound on the Waring rank of the determinant.
\end{remark}

\begin{definition}
Let $h_d \in \mathcal{S}_d^{2d-1}$ be the determinant of a symbolic Hankel matrix (that is, the determinant of the $d \times d$ matrix whose $(i,j)$th entry is the variable $x_{i+j}$).
\end{definition}

\begin{theorem}\label{catupper}
\[A^+(\mathbb{N},d) \le \mathbf{R}(h_d) \le \binom{3d-2}{d} < 6.75^d.\]
\end{theorem}

\begin{proof}
Let $a_1, a_2, \ldots, a_n$ be distinct elements of $\mathbb{R}$, let $A = (a_i^{j-1})_{i \in [n], j \in [d]} \in \mathbb{C}^{d \times n}$, and let $X = \mathrm{diag}(x_1, \ldots, x_n)$. By the Cauchy-Binet formula,
\[\text{det}_d((A \cdot X) \cdot A^T) = g_{A,A} = \sum_{\alpha \in \{0,1\}^n_d} \text{det}_d(A_\alpha A_\alpha)  x^\alpha= \sum_{\alpha \in \{0,1\}^n_d} \text{det}_d(A_\alpha)^2  x^\alpha.\]

Since $A$ is a Vandermonde matrix, $\det_d(A_\alpha)^2 > 0$ for all $\alpha \in \{0,1\}^n_d$. Hence $g_{A,A} \in \mathfrak{E}^+(n,d)$. Now observe that $(A \cdot X) \cdot A^T$ is a Hankel matrix; explicitly, it equals 
\[\sum_{i=1}^n (1,a_i^1, \ldots, a_i^{d-1})^T (1,a_i^1, \ldots, a_i^{d-1}) x_i.\]
Therefore $\text{det}_d(A X A^T) = h_d(A X A^T)$, and so $A^+(\mathbb{N},d) \le \mathbf{R}(h_d)$. Since $h_d$ is a degree-$d$ polynomial in $2d-1$ variables, by the dimension bound of \Cref{gendualc} we have that $\mathbf{R}(h_d) \le  \binom{3d-2}{d}$, and therefore $A^+(\mathbb{N},d) \le \binom{3d-2}{d}$. The theorem follows from Stirling's approximation.
\end{proof}
\begin{remark}\label{bestmax}
The above theorem can be slightly improved by using the state-of-the-art bound \cite{jelisiejew2013upper} on the maximum Waring rank in $\mathcal{S}_d^n$ of
\[\binom{n+d-2}{d-1} - \binom{n+d-6}{d-3},\]
valid when $n,d \ge 3$, which shows that
\[A^+(n,d) \le \mathbf{R}(h_d) \le \binom{3d-3}{d-1} - \binom{3d-7}{d-3}.\]
\end{remark}

It follows from \Cref{derivsopt} that the lower bound on $\mathbf{R}(h_d)$ given by the method of partial derivatives is at most $\binom{\lceil 5d/2 \rceil -1}{\lceil d/2 \rceil } < 3.5^d$. The next theorem shows that the actual lower bound obtained by the method of partial derivatives is exponentially worse than this.

\begin{theorem}\label{catlower}
For all integers $d,u,v>0$ such that $u+v = d$, 
\[\mathrm{rank}(Cat_{h_d}(u,v)) \le \binom{\lceil 3d/2 \rceil }{\lfloor d/2 \rfloor} < 2.6^d.\]
\end{theorem}
\begin{proof}
First note that if $A = \mathrm{Vandermonde}(a_1, \ldots, a_n; d)= (a_i^{j-1}) \in \mathbb{C}^{d \times n}$ with $a_1, \ldots, a_n$ distinct, $g_{A,A}$ equals $h_d$ up to a change of variables. This implies that $\mathrm{rank}(Cat_{h_d}(u,v)) = \mathrm{rank}(Cat_{g_{A,A}}(u,v))$. So we will equivalently work with $f \coloneqq g_{A,A}$. Furthermore we assume that $u \le v$; this is without loss of generality as $Cat_f(u,v) = Cat_f(v,u)^T$. We will then show that $\mathrm{rank}(Cat_f(u,v)) \le m \coloneqq \binom{2v+u}{u}$. As this is maximized when $u = \lfloor d/2 \rfloor$, $v = \lceil d/2 \rceil$, the theorem follows.

The matrix $Cat_f(u,v)$ has rows indexed by monomials $x^{\alpha}$, where $\alpha \in \mathbb{N}^{2d-1}_u$, and columns indexed by monomials $x^\beta$, where $\beta \in \mathbb{N}^{2d-1}_v$. Because $f$ is multilinear, the entries in a row indexed by a non-multilinear monomial $x^\alpha$ will be zero, as $x^\alpha$ annihilates $f$ under differentiation. Similarly, any column indexed by a non-multilinear monomial will have all entries equal to zero. Therefore it suffices to consider the submatrix $M$ of  $Cat_f(u,v)$ indexed by multilinear monomials. We identify the row/column corresponding to $x^\alpha$ with the set $\mathrm{supp}(\alpha) \subseteq [2d-1]$. 

Note that $M_{IJ}$ (the entry of $M$ at row $I$ and column $J$) equals 0 if $I$ and $J$ have a nonempty intersection, and equals $\prod_{i \neq j \in I \cup J} (a_i - a_j)^2$ otherwise. Hence the row indexed by $I$ is a multiple of $\prod_{i \neq j \in I} (a_i - a_j)^2$, and similarly the column indexed by $J$ is a multiple of $\prod_{i \neq j \in J} (a_i - a_j)^2$. Therefore $M = D_1 Q D_2$ for some invertible (diagonal) matrices $D_1$ and $D_2$, and so it suffices to upper bound the rank of $Q$. 

Next, observe that $Q_{IJ} = \prod_{i \in I, j \in J}(a_i - a_j)^2$. Write $I = \{i_1, \ldots, i_u\}$, $J = \{j_1, \ldots, j_v\}$. We now claim that there exist $g_1, h_1, \ldots, g_m , h_m$ with $g_i \in \mathcal{S}^u$, $h_i \in \mathcal{S}^v,$ such that 
\begin{equation}\label{decomp}
Q_{IJ} = \sum_{k=1}^m g_k(a_{i_1}, \ldots, a_{i_u})h_k(a_{j_1}, \ldots, a_{j_v}).
\end{equation}
To see this, view $Q_{IJ}$ as a polynomial in the variables $a_{i_1}, \ldots, a_{i_u}$ with coefficients in $\mathbb{C}[a_{j_1}, \ldots, a_{j_v}]$. This is a symmetric polynomial in $u$ variables, where the maximum degree of any variable in any monomial is $2v$. Therefore $Q_{IJ}$ can be written as in \Cref{decomp} as a sum over symmetrizations of monomials with total degree at most $u$ and maximum individual degree $2v$, for some coefficients $h_k$ in $\mathbb{C}[a_{j_1}, \ldots, a_{j_v}]$. The number of such symmetrizations of monomials is the number of partitions having maximum part size $2v$ and at most $u$ parts, which is $\binom{2v+u}{u} = m$.

Having shown this, it follows that 
\[Q = \sum_{k=1}^m (g_k(a_{i_1} ,\ldots, a_{i_u}))^T_{I \subseteq [2d-1], |I| = u} (h_k(a_{j_1} ,\ldots, a_{j_v}))_{J \subseteq [2d-1], |J| = v},\]
and so $Q$ has rank at most $m$. We conclude by Stirling's approximation.
\end{proof}
\begin{remark}
Numerical evidence suggests that  equality holds in \Cref{catlower} when $u= \lfloor d/2 \rfloor$. This would imply that $\mathbf{R}(h_d) = \Omega(2.59^d)$.
\end{remark}
\subsection{$A(n,d)$ in Positive Characteristic and Abelian Group Algebras}\label{sec3.3}
We briefly introduce a generalization of Waring rank to $\mathcal{S}_d^n(\mathsf{k}) \coloneqq \mathsf{k}[x_1, \ldots, x_n]_d$, where $\mathsf{k}$ is a field of arbitrary characteristic. This notion has been studied extensively as early as 1916 \cite{macaulay1994algebraic}, and directly corresponds to Waring rank in the case that $\mathrm{char}(\mathsf{k}) = 0$. For a thorough algebraic-geometric treatment of this subject, see \cite{iarrobino1999power}. Assume $\mathsf{k}$ is algebraically closed unless stated otherwise. 

\begin{definition}
For $\ell  = \sum_{i=1}^n a_i x_i \in \mathcal{S}_1^n(\mathsf{k})$, let
\[\ell^{[d]} \coloneqq \sum_{\alpha \in \mathbb{N}^n_d} a_1^{\alpha_1} \cdots a_n^{\alpha_n} x^\alpha \in  \mathcal{S}_d^n(\mathsf{k}).\]
\end{definition}
Note that $\ell^{[d]}$ is just $\ell^d$ without any multinomial coefficients. We remark that the projectivization of the set $\{\ell^{[d]} : \ell \in \mathcal{S}^n_1(\mathsf{k})\}$ is the classical \emph{Veronese variety} in algebraic geometry \cite[Corollary A.10]{iarrobino1999power}.
\begin{definition}
For $f \in \mathcal{S}_d^n(\mathsf{k})$, let $\mathbf{R}^\nu(f)$ be the minimum $r$ such that there exist linear forms $\ell_1, \ldots, \ell_r$ with
\[f = \sum_{i=1}^r \ell_i^{[d]},\]
and let 
\[\mathbf{R}_{\mathrm{supp}}^\nu(f) \coloneqq \min (\mathbf{R}^\nu(g) : g \in \mathcal{S}_d^n(\mathsf{k}) , \mathrm{supp}(g) = \mathrm{supp}(f)).\]
\end{definition}
The next proposition shows that the $d = j$ case of \Cref{genduala} holds (ignoring a factorial) with the above definition of rank in the case that $g$ is multilinear. Recall that this fact is key for algorithmic upper bounds. 
\begin{proposition}\label{veroneseworks}
Suppose that $g = \sum_{i=1}^r \ell_i^{[d]} \in \mathcal{S}_d^n$ is multilinear. Then for all $f \in \mathcal{S}_d^n$, 
\[g(\partial \mathbf{x}) f= \sum_{i=1}^r f(\ell_i^*).\]
\end{proposition}
\begin{proof}
Suppose that $g = \sum_\alpha b_\alpha x^\alpha$ and $\ell_i = (\sum_{j=1}^n c_{i,j} x_j)^{[d]}$. Note that $b_\alpha = \sum_{i=1}^r c_{i,1}^{\alpha_1} \cdots c_{i,n}^{\alpha_n}$. If $f = \sum_\alpha a_\alpha x^\alpha$, then since $g$ is multilinear, $g(\partial \mathbf{x}) f = \sum_\alpha a_\alpha b_\alpha$. On the other hand,
\[\sum_{i=1}^r f(c_{i,1}, \ldots, c_{i,n}) = \sum_{i=1}^r \sum_{\alpha} a_\alpha c_{i,1}^{\alpha_1} \cdots c_{i,n}^{\alpha_n} = \sum_\alpha a_\alpha b_\alpha.\qedhere\]
\end{proof}

\begin{definition}\label{genand}
Let $A_\mathsf{k}(n,d) \coloneqq \mathbf{R}_{\mathrm{supp}}^\nu(e_{n,d})$.
\end{definition}
It is easy to see that if $\mathsf{k} = \mathbb{C}$ and if $g$ is multilinear, $\mathbf{R}(g) = \mathbf{R}^\nu(g)$. This implies that $A_\mathbb{C}(n,d)=A(n,d)$, and so the above definition really does generalize $A(n,d)$.
\begin{theorem}\label{generallower}
For all $n \ge d$, $A_\mathsf{k}(n,d) \ge 2^{d-1}$.
\end{theorem}
\begin{proof}
It follows from an argument identical to that of \Cref{trivialboundsa} that $A(d,d) \le A(n,d)$ for all $n \ge d$. As it was shown in \cite{ranestad2011rank} that $\mathbf{R}^\nu(x_1 \cdots x_d) \ge 2^{d-1}$, the conclusion follows.
\end{proof}

\begin{definition}
Given $A \in \mathsf{k}^{d \times n}$, let
\begin{equation}
g_A \coloneqq \sum_{\alpha \in \mathbb{N}^n_d} \text{per}_d(A_\alpha) x^\alpha.
\end{equation}
\end{definition}

\begin{lemma}\label{persum}
Let $\mathsf{k}$ be arbitrary and let $A \in \mathsf{k}^{d \times n}$. Then $\mathbf{R}^\nu(g_A) \le 2^d - 1$.
\end{lemma}
\begin{proof}
For $1 \le i \le d$, let $L_i \coloneqq \sum_{j =1}^n A_{ij} y_j \in \mathsf{k}[y_1, \ldots, y_n]$. Now consider 
\[\sum_{\alpha \in \mathbb{N}^n_d} L_1^{\alpha_1} \cdots L_n^{\alpha_n} x^\alpha \in \mathsf{k}[y_1, \ldots, y_n][x_1, \ldots,x_n].\]
 Note that the coefficient of $y_1 \cdots y_d$ in this polynomial is equal to $g_A$.  It then follows from inclusion-exclusion (or \Cref{ryser}) that this coefficient equals
\begin{equation}
\sum_{\alpha \in \{0,1\}^d} (-1)^{|\alpha|+d} (\sum_{i=1}^n x_i \sum_{j=1}^d  \alpha_j A_{i,j})^{[d]}.\label{perdecomp}\qedhere
\end{equation}
\end{proof}

\begin{theorem}\label{char2ez}
If $\mathsf{k}$ is infinite and $\mathrm{char}(\mathsf{k}) = 2$, $A_\mathsf{k}(n,d) \le 2^d-1$.
\end{theorem}
\begin{proof}
Let $A \in \mathsf{k}^{d \times n}$ be a matrix with non-vanishing $d \times d$ minors. Since $\mathrm{char}(\mathsf{k}) = 2$,
\[g_A = \sum_{\alpha \in \mathbb{N}^n_d} \text{det}_d(A_\alpha) x^\alpha.\]
If $\alpha \notin \{0,1\}^n$ then $A_\alpha$ has a repeated column and so $\det(A_\alpha) = 0$. Otherwise $\det(A_\alpha) \neq 0$. Therefore $g_A$ has the desired support. The conclusion follows from \Cref{persum}.
\end{proof}

\Cref{char2ez} gives the following $2^d \mathrm{poly}(n)$-time algorithm for testing if a polynomial $f \in \mathcal{S}_d^n(\mathsf{k})$ over a large enough field of characteristic 2 is supported on any multilinear monomial. For $U \subseteq \mathsf{k}$, where $|U| \ge 2d$, choose $a = (a_1, \ldots, a_n) \in U^n$ uniformly at random, and take $A \in \mathsf{k}^{d \times n}$ to have nonvanishing $d \times d$ minors. Then compute
\begin{equation}\label{c2algc}
\sum_{\alpha \in \{0,1\}^d} f(a_1 \sum_{j=1}^d  \alpha_j A_{1,j}, \ldots, a_n  \sum_{j=1}^d  \alpha_j A_{n,j}).
\end{equation}
It follows from \Cref{veroneseworks}, \Cref{char2ez}, and the Schwartz-Zippel lemma that this quantity is nonzero with probability at least $1/2$ when $f$ is supported on a multilinenar monomial, and zero otherwise. If $f = \sum_\alpha b_\alpha x^\alpha$, this algorithm computes

\[\sum_{\alpha \in \{0,1\}^n_d} b_\alpha a^\alpha \det(A_\alpha).\]

The ``option 2'' implementation of ``decide-multilinear" in \cite{koutis2008faster} is obtained exactly if instead we choose $A \in \mathbb{Z}_2^{d \times n}$ uniformly at random and take $a_1, \ldots, a_n = 1$. Similarly, the algorithm of \cite{williams2009finding} is obtained by choosing both $A \in \mathbb{Z}_2^{d \times n}$ and $a_1, \ldots, a_n \in \mathsf{k}$ uniformly at random. Additionally, the algorithm of \cite{bjorklund2010determinant} for detecting Hamiltonian cycles reduces to computing \Cref{c2algc} where $a_1, \ldots,  a_n =  1$, $A \in \mathsf{k}^{d \times n}$ is chosen uniformly at random, and the generating polynomial $f$ has the property that $\deg f \approx 3d/4$. This explains the relevance of ``determinant sums'' to \cite{bjorklund2010determinant} and shows that \cite{williams2009finding,koutis2008faster} were in fact also computing ``determinant sums''. This connection was made earlier in \cite{Brand2018Extensorcoding}.

The algorithms of \cite{williams2009finding,koutis2008faster} were presented in terms of a property of abelian group algebras. The following theorem elucidates the connection between support rank and this property.

\begin{theorem}\label{groupalg}
Let $G$ be an abelian group, and let $y_1, \ldots, y_n \in \mathsf{k}[G]$. For $\alpha \in \mathbb{N}^n$, let $f_\alpha \coloneqq \prod_{i=1}^n y_i^{\alpha_i}$. Define
\[T \coloneqq \{\alpha \in \mathbb{N}^n_d : f_\alpha(\mathrm{Id}_G) \neq 0\}.\]
Then $\mathbf{R}_{\mathrm{supp}}^\nu(\sum_{\alpha \in T} x^\alpha) \le |G|$.
\end{theorem}
\begin{proof}
Let $\rho$ be the regular representation of $G$; this extends linearly to a representation of $\mathsf{k}[G]$. Consider the  $|G| \times |G|$  matrices $\rho(y_1) , \ldots, \rho(y_n)$. Since $G$ is abelian, there exists an invertible matrix $A$ so that $\rho(y_i) = A \Lambda_i A^{-1}$ for all $i \in [n]$ and some diagonal matrices $\Lambda_1, \ldots, \Lambda_n$. 

By assumption, we have that for all $\alpha \in \mathbb{N}^n_d$, $f_\alpha(\mathrm{Id}_G) \neq 0$ if and only if $\alpha \in T$. Note that $f_\alpha(\mathrm{Id}_G) \neq 0$ if and only if for some $\lambda \neq 0$ and all $i \in |G|$,  $\rho(f_\alpha)_{i,i} = \lambda$. Letting $D \in \mathsf{k}^{|G| \times |G|}$ be a diagonal matrix with nonzero trace, it follows that $\mathrm{tr}(D\cdot\rho(f_\alpha)) \neq 0$ if and only if $\alpha \in T$. Note that
\begin{align*}
\mathrm{tr}(D \cdot \rho(f_\alpha)) &= \mathrm{tr}(D \cdot \rho(\prod_{i=1}^n y_i^{\alpha_i})) = \mathrm{tr}(D \cdot \prod_{i=1}^n \rho(y_i)^{\alpha_i}),\\
&= \mathrm{tr}(D \cdot \prod_{i=1}^n (A \Lambda_i A^{-1})^{\alpha_i}),\\
&= \mathrm{tr}(D \cdot \prod_{i=1}^n \Lambda_i^{\alpha_i}).
\end{align*}
Let $M_i\coloneqq D^{1/n} \Lambda_i$. By the above discussion, for all $\alpha \in \mathbb{N}^n_d$, $\mathrm{tr}(\prod_{i=1}^n M_i^{\alpha_i}) \neq 0$ if and only if $\alpha \in T$. 

Define the linear forms $\ell_i = \sum_{j=1}^n (M_j)_{i,i} x_i$ for all $i \in |G|$. We now claim that $P \coloneqq \sum_{i=1}^n \ell_i^{[d]}$ has the desired support. To see this, consider the coefficient of $x^\alpha$ in $P$, where $|\alpha| = d$. By definition, this is equal to
\[\sum_{i=1}^{|G|} (M_1)_{i,i}^{\alpha_1} \cdots (M_n)_{i,i}^{\alpha_n} = \mathrm{tr}(\prod_{i=1}^n M_i^{\alpha_i}),\]
and hence the claim holds.
\end{proof}

\Cref{groupalg} allows to to recover the approach of \cite{koutis2008faster,williams2009finding} from a support-rank perspective. Let $G = \mathbb{Z}_2^d$, and let $v_1, \ldots, v_n \in G$ be chosen independently and random. Then let $y_i \coloneqq \mathrm{Id}_G + v_i \in \mathsf{k}[G]$ for all $i$ in the statement of \Cref{groupalg}. The key fact used in \cite{koutis2008faster,williams2009finding} was that when char($\mathsf{k}) = 2$, $f_\alpha(\mathrm{Id}_G) = 0$ whenever $\alpha \notin \{0,1\}^n_d$, and for any $\alpha \in \{0,1\}^n_d$, $f_\alpha(\mathrm{Id}_G) \neq 0$ with probability at least $1/4$. The algorithms of \cite{koutis2008faster,williams2009finding} then follow by using the decomposition given by \Cref{groupalg}. Note that this algorithm does not use a decomposition of a multilinear polynomial supported on all multilinear monomials, but rather it samples a multilinear polynomial that is supported on a given multilinear monomial with \emph{constant probability}.

\subsection{A Recursive Approach for Bounding $A(n,d)$}\label{sec4}
In this section we provide a recursive method for upper bounding $A^+(n,d)$ and $A^\varepsilon(n,d)$. We will start with a recursive bound on $A^\varepsilon(n,d)$ for varying $n$ and fixed $d$, and later build upon this to give a recursive bound on $A^+(n,d)$ for all $n$ and $d$. 

\subsubsection{A Recursive Bound on $A^\varepsilon(n,d)$ for Fixed $d$}\label{sec4.1}
We will first need the following tool introduced in \cite{alon2007balanced}.

 \begin{definition}
 For $\delta > 1$, a $\delta$-balanced $(n,k,l)$-splitter $\mathcal{F}$ is a family of functions from $[n]$ to $[l]$ such that for some real number $c$, for all $S \subseteq[n]$ where $|S| = k$, the number of functions in $\mathcal{F}$ that are injective on $S$ is between $c/\delta$ and $c\delta$. 
 
A $\delta$-balanced $(n,k,k)$-splitter will be called a $\delta$-balanced $(n,k)$-perfect hash family. If $\mathcal{F}$ only satisfies the property that for each $S \subseteq [n]$, where $|S| = k$, there exists \emph{some} function in $\mathcal{F}$ that is injective on $S$, we call $\mathcal{F}$ an $(n,k,l)$-splitter.
 \end{definition}
 
 The next fact essentially appears in \cite{alon2007balanced}; we reproduce the proof for completeness. Here $(n)_k \coloneqq n(n-1) \cdots (n-k+1)$ denotes the falling factorial.
 \begin{lemma}\label{splitterbd}
For $1<\delta \le 2$, there exists a $\delta$-balanced $(n,k,l)$-splitter of size
 \[O \left (\frac{l^k \cdot  k \log n}{(l)_k (\delta-1)^2} \right) .\]
 \end{lemma}
\begin{proof}
Set $p \coloneqq \frac{(l)_k}{l^k}$ and $M \coloneqq \lceil \frac{8(k \log n +1)}{p(\delta-1)^2} \rceil$. Choose $M$ independent random functions from $[n]$ to $[l]$. For any $S \subseteq [n]$ of size $k$, the expected number of functions that are injective on $S$ is $pM$. By the Chernoff bound, the probability that the number of functions that are injective on $S$ is less than $pM/\delta$ or greater than $pM\delta$ is at most $2e^{-(\delta-1)^2 pM/8}$. Then by a union bound the expected number of such sets for which the number of 1-1 functions is not as desired is at most 
\[\binom{n}{k} 2e^{-(\delta-1)^2 pM/8} \le \binom{n}{k} 2e^{-(k \log n + 1)}<1.\qedhere\]
\end{proof}
 
\begin{theorem}\label{aepsilonbound}
Suppose $f \in \mathfrak{E}^{\varepsilon_0}(n_0,d)$ where $0 < \varepsilon_0 < 1$. Then for all $\varepsilon_0 < \varepsilon <1$ and all $n \ge d$,
\[A^\varepsilon(n,d) \le O\left (\frac{\mathbf{R}(f) \cdot n_0^d \cdot  d \log n}{(n_0)_d (\delta-1)^2} \right),\]
where $\delta \coloneqq \min(\frac{1+\varepsilon}{1+\varepsilon_0}, \frac{1-\varepsilon_0}{1-\varepsilon})$.
\end{theorem}
\begin{proof}
If $n \le n_0$ the theorem follows from \Cref{trivialboundsa}. Hence we will assume that $n > n_0$.

Let $\mathcal{F} = \{\pi_i : i \in [M]\}$ be a $\delta$-balanced $(n,d,n_0)$-splitter of minimal size $M$. For all $(i,j) \in [M]\times[n_0]$, define the linear forms $L_{i,j} = \sum_{k \in \pi_i^{-1}(j)} x_k$. Now we claim that for some constant $c$,
\[f' \coloneqq \frac{1}{c}  \sum_{i=1}^M f(L_{i,1},L_{i,2}, \ldots, L_{i,n_0}) \in A^{\varepsilon}(n,d).\]
First notice that since $f$ is multilinear and $L_{i,1}, \ldots, L_{i,n_0}$ are linear forms with disjoint supports for all $i$, $f'$ is also multilinear.  Next, by virtue of the fact that $f\in \mathfrak{E}^{\varepsilon_0}(n_0,d)$, the coefficient of any multilinear monomial $x^\alpha$ in $f(L_{i,1},L_{i,2}, \ldots, L_{i,n_0})$ is in the range $[1-\varepsilon_0,1+\varepsilon_0]$ if and only if $\pi_i$ is injective on $\mathrm{supp}(\alpha)$. Then because $\mathcal{F}$ is a $\delta$-balanced splitter, there are between $c/\delta$ and $c\delta$ such contributions to the coefficient of $x^\alpha$ in the above sum, for some fixed real number $c$. But this implies that the coefficient of $x^\alpha$ in $f'$ is between $(1-\varepsilon_0)/\delta$ and $(1+\varepsilon_0)\delta$, which by our choice of $\delta$ implies that $f \in \mathfrak{E}^{\varepsilon}(n,d)$. By subadditivity of rank, $\mathbf{R}(f') \le M \cdot \mathbf{R}(f)$, and the theorem follows by the bound on $M$ given by \Cref{splitterbd}.
\end{proof}

\begin{remark}
As Waring rank can be strictly subadditive, it is possible that the final step of the above lemma is far from optimal; see also \Cref{strassen}.
\end{remark}

\begin{theorem}\label{pracub}
For all $0<\varepsilon <1$, $A^\varepsilon(n,d) \le O(4.075^d \varepsilon^{-2} \log n)$.
\end{theorem}

\begin{proof}
Let $c \ge 1$ be a constant to be determined later. Taking $n_0 =\lceil cd \rceil$,  $f = e_{n_0,d}$, $\varepsilon_0 = \varepsilon/2$ in \Cref{aepsilonbound},
\[A^\varepsilon(n,d) \le O\left (\frac{\mathbf{R}(e_{n_0, d}) \cdot n_0^d \cdot  d \log n}{(n_0)_d (\delta-1)^2} \right)\]
where $\delta = \min(\frac{1+\varepsilon}{1+\varepsilon/2}, \frac{1-\varepsilon/2}{1-\varepsilon}) = \frac{1+\varepsilon}{1+\varepsilon/2} \ge \varepsilon/3+1$. Combining this with the upper bound on $\mathbf{R}(e_{n_0,d})$ given in \cite{lee2016power},
\begin{align*}
A^\varepsilon(n,d) &\le O\left (\binom{n_0}{\lfloor d/2 \rfloor} \frac{n_0^d}{(n_0)_{d}} \varepsilon^{-2} d^2 \log n \right )=O\left (\frac{\lceil cd \rceil^d}{\lfloor d/2 \rfloor!} \frac{\lceil d(c-1) \rceil!}{\lceil d(c-1/2) \rceil!}\varepsilon^{-2}d^2 \log n \right ).
\intertext{Applying Stirling's inequality,}
A^\varepsilon(n,d) &\le O\left (\left ( \frac{cd}{\sqrt{d/(2e)}} \right )^d \left (\frac{d(c-1)}{e} \right)^{d(c-1)} \left (\frac{e}{d(c-1/2)} \varepsilon^{-2} \right )^{d(c-1/2)}\varepsilon^{-2} d^2 \log n \right )\\
&= O\left (\left (\sqrt{2e} \cdot c \left (\frac{c-1}{e} \right )^{c-1} \left (\frac{e}{c-1/2} \right )^{c-1/2} \right )^d \varepsilon^{-2} d^2 \log n \right ).
\end{align*}
Using a computer we found that this is minimized when $c \approx 1.55$, in which case we obtain an upper bound of $O(4.075^d \varepsilon^{-2} \log n)$.
\end{proof}

\begin{remark}\label{colorred}
If we take $f = x_1 x_2 \cdots x_d$ and use the upper bound on $\mathbf{R}(x_1 \cdots x_d)$ given by \Cref{ryser}, it follows from \Cref{aepsilonbound} that
\[A^\varepsilon(n,d) \le (2^d-1) \frac{d^d}{d!} \varepsilon^{-2} = O((2e)^d\varepsilon^{-2} ) = O(5.44^d \cdot \varepsilon^{-2} ).\]

The decomposition implicit in the above bound is as follows. Let $\mathcal{F}$ be an $(1+\varepsilon)$-balanced $(n,d)$-perfect hash family. For $\pi \in \mathcal{F}$ and $i \in [d]$, let $L_{\pi,i} \coloneqq \sum_{j \in \pi^{-1}(i)} x_j$. Then for some $c >0$,
\[\frac{1}{c} \sum_{\pi \in \mathcal{F}} \sum_{\substack{\alpha \in \{0,1\}^{d}}} (-1)^{|\alpha|+d}\left ( \sum_{i=1}^{d} \alpha_i L_{\pi,i} \right )^d \in \mathfrak{E}^\varepsilon(n,d).\]
Applying this to the cycle-generating polynomial \Cref{gencycles}, one finds that a $(1\pm \varepsilon)$-approximation of the number of length-d cycles in the graph $G$ is given by
\[\frac{1}{c \cdot d!} \sum_{\pi \in \mathcal{F}} \sum_{\substack{\alpha \in \{0,1\}^{d}}} (-1)^{|\alpha|+d} f_G(\alpha_{\pi(1)}, \ldots, \alpha_{\pi(n)}).\]
This is equivalent to the color-coding algorithm for counting cycles described in \cite{alon2009balanced}, except we use inclusion-exclusion instead of dynamic programming to count the number of colorful simple cycles for a given coloring. Similarly, by replacing $\mathcal{F}$ with an $(n,d)$-perfect hash family one obtains an algorithm for detecting simple cycles that parallels the one given in \cite{alon1995color}. We note that using inclusion-exclusion rather than dynamic programming reduces the space complexity of the counting step from exponential to polynomial.

Furthermore, this bound is naturally derived by an application of color-coding. Using each function in a $(1+\varepsilon)$-balanced $(n,d)$-perfect hash family we color the variables $x_1, \ldots, x_n$ using $d$ colors. To each color we associate the linear form equal to the sum of the variables of that color. Since these linear forms have disjoint support, their product is multilinear. Summing the resulting products of linear forms for each function in the family, any given multilinear monomial appears with coefficient between $c/(1+\varepsilon)$ and $c(1+\varepsilon)$. The resulting polynomial is a sum of products of $|\mathcal{F}|$ linear forms, which can be written as a sum of powers of $O(|\mathcal{F}|2^d)$ linear forms using \Cref{ryser}.

An improvement to color-coding was made in \cite{huffner2008algorithm} based on the idea of using $n_0 \coloneqq \lceil 1.3d \rceil$ colors rather than $d$. We recover this result as follows. By applying \Cref{aepsilonbound} with $f=e_{n_0,d}$ and using the suboptimal bound on $\mathbf{R}(e_{n_0,d})$ given by \Cref{ryser},
\[A^+(n,d) \le O\left (\binom{1.3d}{d} \frac{(1.3d)^d}{(1.3d)_d} d \log n \right ) = O(4.32^d \log n).\]
In fact, the choice of $n_0 = \lceil 1.3d \rceil$ is optimal if we are using the rank bound of \Cref{ryser}; this follows from the same calculation done in \cite[Section 8]{gutin2018designing}. The algorithm resulting from this bound was virtually described in \cite{gutin2018designing,amini2009counting}.
\end{remark}
\subsubsection{A Recursive Bound on $A^+(n,d)$ for all $n$ and $d$}\label{sec4.2}
\begin{definition}
For $g \in \mathcal{S}_d^n$ and $s,t \in \mathbb{N}$, let
$$g^{\circledast (s,t)} \coloneqq \sum_{i=1}^s \prod_{j=1}^t g(x_{i,j,1}, x_{i,j,2}, \ldots , x_{i,j,n}) \in \mathbb{C}[x_{i,j,k} : (i,j,k) \in [s] \times [t] \times [n]].$$
\end{definition}
In words, $g^{\circledast (s,t)}$ is obtained from $g$ by taking the $t$-fold product of $g$ with itself using disjoint sets of variables, and then taking the $s$-fold sum of the resulting polynomial using disjoint sets of variables.

\begin{lemma}\label{cdastbd}
For all $g \in \mathcal{S}_d^n$, $\mathbf{R}(g^{\circledast(s,t)}) \le s ((d+1) \mathbf{R}(g))^t$.
\end{lemma}

\begin{proof}
By subadditivity of Waring rank, $\mathbf{R}(g^{\circledast(s,t)}) \le s \mathbf{R}(g^{\circledast(1,t)})$. Now letting $r = \mathbf{R}(g)$, there exist linear forms $\ell_{i,j} \in \mathbb{C}[x_{1,i,1}, \ldots, x_{1,i,n}]$ for $(i,j) \in [t] \times [r]$ so that
\[g^{\circledast(1,t)} = \prod_{i=1}^t \sum_{j=1}^r \ell_{i,j}^d =  \sum_{v \in [r]^t} \prod_{i=1}^t \ell_{i, v_i}^{d}.\]
Using the fact that $\mathbf{R}(\prod_{i=1}^t x_i^d) \le (d+1)^t$ (which follows from e.g. \Cref{ryser}\footnote{The slightly better bound of $(d+1)^{t-1}$ given in \cite{ranestad2011rank} can be used here.}), it follows that $\mathbf{R}(g^{\circledast(s,t)}) \le s \mathbf{R}(g^{\circledast(1,t)}) \le s ((d+1)\mathbf{R}(g))^t$.
\end{proof}
\begin{remark}\label{strassen}
The first step of the above lemma is to apply subadditivity of Waring rank to polynomials in disjoint sets of variables. \emph{Strassen's direct sum conjecture} claims that rank is actually additive in this case; see \cite{carlini2015progress} for more. It was recently shown in \cite{shitov2017counterexample} that the tensor version of this conjecture is false; if the polynomial version is also false, the upper bound of \Cref{cdastbd} may not be optimal.
\end{remark}

\begin{definition}
An $(n,d,n_0,d_0)$-perfect splitter, where $n \ge d$, $n_0 \ge d_0$, and $d_0 \mid d$, is a family of functions $\mathcal{F} = \{ \pi : [n] \to [d/d_0] \times [n_0]\}$ such that for all $S \subseteq [n]$ where $|S| = d$, there exists a $\pi \in \mathcal{F}$ such that for all $i \in [d/d_0]$, $\pi(S)$ contains $d_0$ elements whose first coordinate is $i$, and any two elements in $\pi(S)$ with the same first coordinate have differing second coordinates.
\end{definition}
In other words, we want the elements of $\pi(S)$ to be ``split evenly'' by their first coordinate, and those elements with the same first coordinate should have different second coordinates. As special cases, an $(n,d,d,d)$-perfect splitter is a $(n,d)$-perfect hash family, and when $n_0 \ge n$, an $(n,d,n_0,d_0)$-perfect splitter is a $(n,d,d_0)$-splitter.
\begin{definition}
For $n \ge d$, $n_0 \ge d_0$, and $d_0 \mid d$, let
\[\sigma(n,d,n_0,d_0) \coloneqq \left \lceil \left ( \frac{n_0^{d_0}}{(n_0)_{d_0}} \right )^{d/d_0} \frac{d_0!^{d/d_0} (d/d_0)^{d}}{d!} d \log n\right \rceil.\]
\end{definition}
\begin{proposition}\label{pslowerbd}
There exists an $(n,d,n_0,d_0)$-perfect splitter of size $\sigma(n,d,n_0,d_0)$.
\end{proposition}
\begin{proof}
We will consider the probability that a random function $\pi$ has the desired effect on a fixed subset $S \subseteq [n],$ where $|S| = d$. The conclusion will then follow from a union bound.

Let $\pi : [n] \to [d/d_0] \times [n_0]$ be chosen uniformly at random. The probability that each integer in $[d/d_0]$ appears equally often as the first coordinate in $\pi(S)$ equals
\[p_1 \coloneqq \frac{d!}{d_0!^{d/d_0} (d/d_0)^{d}}.\]
Assuming this happens, the probability that all elements in $\pi(S)$ with a given first coordinate are assigned different second coordinates equals
\[p_2 \coloneqq \frac{(n_0)_{d_0}}{n_0^{d_0}},\]
and so with probability $p_2^{d/d_0}$ this happens for all $d/d_0$ choices of the first coordinate. Hence if we generate $c = \lceil (p_1 p_2^{d/d_0})^{-1} \rceil$ independent and uniformly random functions, some function has the desired effect on $S$ with probability at least $1-e^{-1}$. Therefore if we generate $\lceil cd \log n \rceil$ random functions, the expected number of subsets for which no function has the desired effect on equals 
\[\binom{n}{d} e^{-\lceil d \log n \rceil} < 1.\qedhere\]
\end{proof}

\begin{theorem}\label{masterthm}
Let $f \in \mathfrak{E}^+(n_0,d_0)$. Then for all integers $n,d$ where $n \ge d$,
\[A^+(n,d) \le s ((d_0+1) \mathbf{R}(f))^{\lceil d/d_0 \rceil},\]
where 
\[s = \sigma(n+\lceil d/d_0 \rceil d_0 - d, \lceil d/d_0 \rceil d_0,n_0,d_0).\]
\end{theorem}
\begin{proof}
 We start with the case that $d = t \cdot d_0$ for some $t \in \mathbb{N}$. Let $\mathcal{F} = \{\pi_i : i \in [s]\}$ be an $(n,d,n_0,d_0)$-perfect splitter of minimal size.  For $(i,j,k) \in [s] \times [t] \times [n_0]$, let $L_{i,j,k} \coloneqq \sum_{m \in \pi_i^{-1}(j,k)} x_m$. We now claim that $g^{\circledast(s,t)}(L_{i,j,k}) \in \mathfrak{E}^+(n,d)$. To see this, first note that for any $i$, the linear forms $\{L_{i,j,k}: (j,k) \in [t] \times [n_0]\}$ have disjoint support. Since $f$ is multilinear, it follows that 
\[f_i \coloneqq f(L_{i,1,1}, \ldots, L_{i,1,n_0}) \cdots f(L_{i,t,1}, \ldots, L_{i,t,n_0})\]
is multilinear for all $i$, and therefore so is $f^{\circledast(s,t)}(L_{i,j,k})$.
 
Now consider the coefficient of some degree-$d$ multilinear monomial $x^{\alpha}$ in $f_i$. Since $f$ has nonnegative coefficients, this will be nonnegative. Furthermore, if $\pi_i$ splits the set $\mathrm{supp}(\alpha)$ evenly by first coordinate and all elements in $\pi_i(\mathrm{supp}(\alpha))$ with the same first coordinates have different coordinates, this coefficient will be strictly positive by definition of the linear forms $L_{i,j,k}$. Since $\mathcal{F}$ is a perfect splitter, each degree-$d$ multilinear monomial will then appear with a positive coefficient. Therefore by \Cref{pslowerbd},
\[A^+(n,d) \le \mathbf{R}(f^{\circledast(s,t)}) \le s ((d_0+1) \mathbf{R}(f))^{d/d_0}.\] 

Now suppose that $d_0 \nmid d$. By \Cref{trivialboundsb}, we have that 
\[A^+(n,d) \le A^+(n+ \lceil d/d_0 \rceil d_0 - d, \lceil d/d_0 \rceil d_0),\]
which is at most $s ((d_0+1) \mathbf{R}(f))^{\lceil d/d_0 \rceil}$ by a reduction to the case when $d_0 \mid d$.
\end{proof}

Note that by taking $d_0 = d$ in the above theorem, we find that 
\[A^+(n,d) \le O\left (\frac{A^+(n_0,d) \cdot n_0^d \cdot  d \log n}{(n_0)_d} \right),\]
recovering \Cref{aepsilonbound} in the case of nonnegative support rank.

\begin{example}\label{d410}
\Cref{masterthm} suggests bounding $A^+(\mathbb{N},d)$ for small values of $d$ as an approach to improve the upper bounds of this section. For example, suppose that $A^+(\mathbb{N},4) \le 10$. Then we have that for all $n_0 \ge 4$ and all $n,d$,
\begin{align*}
A(n,d) &\le \sigma(n+ 4\lceil d/4 \rceil - d, \lceil d/4 \rceil 4,n_0,4) 5^{\lceil d/4 \rceil-1} 10^{\lceil d/4 \rceil}\\
&= O\left (\left (\frac{n_0^4}{{n_0}_{(4)})} \right )^{d/4} \frac{4!^{d/4} (d/4)^d}{d!} \log \binom{n}{d}  50^{d/4} \right )\\
&= O\left (\left (\frac{n_0^4}{{n_0}_{(4)})} \right )^{d/4} (e \cdot 1200^{1/4}/4)^d d \log n \right )\\
&= O\left (\left (\frac{n_0^4}{{n_0}_{(4)})} \right )^{d/4} 3.9998^d d \log n \right).
\end{align*}
Taking $n_0 \ge 33700$, we conclude that $A(n,d) \le O( 3.9999^d \log n)$. 

In contrast, the best upper bound we know on $A^+(\mathbb{N},4)$ is 79, which follows from \Cref{bestmax}. When used in \Cref{masterthm} this only shows that $A^+(n,d) \le O(6.706^d \log n)$.
\end{example}

\section{Applications}\label{sec4}

We first give a proof of \Cref{appl1}.

\begin{proof}[Proof of \Cref{appl1}]
Set $n_0 \coloneqq \lceil 1.55d \rceil$, $p \coloneqq (n_0)_d/n_0^d$, and $M \coloneqq \lceil3\varepsilon^{-2}/p \rceil$. Let $\mathcal{F} \coloneqq \{\pi_i : [n] \to [n_0], i \in [M]\}$ be a family of independent and uniformly random functions and let $L_{\pi_i,j} \coloneqq \sum_{j \in \pi_i^{-1}(j)} x_j$. The algorithm is to compute 
\[\frac{1}{pM} \sum_{i \in [M]} e_{n_0,d}(L_{\pi_i,1}(\partial \mathbf{x}), \ldots, L_{\pi_i,n_0}(\partial \mathbf{x}))f.\]
This can be rewritten in terms of evaluations of $f$ as follows. For $S \subseteq [n_0]$ and $i \in [n_0]$, let $\delta_{S,i} = -1$ if $i \in S$ and $\delta_{S,i}= 1$ otherwise. Then by \Cref{genduala} and the upper bound on $e_{n_0,d}$ given in \cite{lee2016power}, for $d$ odd this is equal to
 \[\frac{1}{p M \cdot 2^{d-1}}\sum_{i \in [M]} \sum_{\substack{S \subset [n_0]\\ |S| \le \lfloor d/2 \rfloor}} (-1)^{|S|} \binom{n_0-\lfloor d/2 \rfloor - |S| -1}{\lfloor d/2 \rfloor - |S|} f(\delta_{S,\pi_i(1)}, \ldots, \delta_{S,\pi_i(n)}),\]
and for $d$ even is equal to
 \[\frac{1}{p M \cdot 2^{d-1} (n_0-d) }\sum_{i \in [M]} \sum_{\substack{S \subset [n_0]\\ |S| \le d/2 }} (-1)^{|S|} \binom{n_0-d/2  - |S| -1}{d/2 - |S|}(n_0-2|S|) f(\delta_{S,\pi_i(1)}, \ldots, \delta_{S,\pi_i(n)}).\]
 
Note that this algorithm makes $M \sum_{i=0}^{\lfloor d/2 \rfloor} \binom{n_0}{i}$ evaluations of $f$, which from the same calculation of \Cref{pracub} is at most $O(4.075^d \varepsilon^{-2})$.  The stated time and space bounds then follow from the straightforward evaluation of these expressions.

We now show that this algorithm is correct. Write $f = \sum_{\alpha \in \mathbb{N}^n} a_\alpha x^\alpha$ and let $\pi : [n] \to [n_0]$ be chosen uniformly at random. Define the linear forms $L_i \coloneqq \sum_{j \in \pi^{-1}(i)} x_j$ for all $i \in [n]$, and write $e_{n_0,d}(L_1, \ldots, L_{n_0}) = \sum_{\alpha \in \{0,1\}^n_d} b_\alpha x^\alpha$. Let $Y_\pi \coloneqq e_{n_0,d}(L_1(\partial \mathbf{x}), \ldots, L_{n_0}(\partial \mathbf{x})) f = \sum_{\alpha \in \{0,1\}^n} a_\alpha b_\alpha$.

First observe that for any fixed $\alpha \in \{0,1\}^n_d$, $b_\alpha = 1$ with probability $p$, and $b_\alpha = 0$ with probability $1-p$. By linearity of expectation, it follows that $\mathbb{E}[Y_\pi] = p \cdot e_{n,d}(\partial \mathbf{x})f$. Moreover,
\[\mathrm{Var}[Y_\pi] = \sum_{\alpha} \mathrm{Var}[a_\alpha b_\alpha] + \sum_{\beta \neq \alpha} \mathrm{Cov}[a_\alpha b_\alpha,a_\beta b_\beta] = \sum_{\alpha} a_\alpha^2 \mathrm{Var}[b_\alpha] + \sum_{\beta \neq \alpha} a_\alpha a_\beta \mathrm{Cov}[b_\alpha,b_\beta].\]
As the probability that $b_\alpha = b_\beta = 1$ is at most $p$ for all $\alpha,\beta$, we have that 
\[\mathrm{Cov}[b_\alpha,b_\beta] = \mathbb{E}[b_\alpha b_\beta] - \mathbb{E}[b_\alpha]\mathbb{E}[b_\beta] \le p,\]
and hence $\mathrm{Var}[Y_\pi] \le p (e_{n,d}(\partial \mathbf{x})f)^2$.

If we repeat this process $M$ times, choosing $M$ independent and random functions $\pi_1, \ldots, \pi_M$ and computing $Z \coloneqq \frac{1}{M}(Y_{\pi_1}+ \cdots + Y_{\pi_M})$, then $\mathbb{E}[Z] = p \cdot e_{n,d}(\partial \mathbf{x})f$ and $Var[Z] = Var[Y_\pi]/M \le p \cdot (e_{n,d}(\partial \mathbf{x})f)^2/M$. By Chebychev's inequality, the probability that $Z$ is smaller or bigger than its expectation by $\varepsilon p e_{n,d}(\partial \mathbf{x})f$ is at most $\varepsilon^{-2}/pM$, which by choice of $M$ is at most $1/3$. Dividing by $p$ we obtained the desired approximation. Note that this is exactly the algorithm described above.
\end{proof}
\begin{remark}
In order to derandomize \Cref{appl1}, it would suffice to give a near-optimal construction of a $(1+\varepsilon)$-balanced $(n,d,1.55d)$-splitter, as first defined in \cite{alon2007balanced}. We note that such a construction was given for (``unbalanced'') $(n,k,\alpha k)$-splitters for all $\alpha \ge 1$ in \cite{gutin2018designing}. Furthermore, note that for any \emph{fixed} values of $n$ and $d$, \Cref{appl1} can be made deterministic by taking $\mathcal{F}$ to be a $(1+\varepsilon)$-balanced $(n,d,1.55d)$-splitter of optimal size.
\end{remark}
\subsection{Counting Subgraphs of Bounded Treewidth}
We now prove \Cref{twcount}.
\begin{definition}
For graphs $G,H$, where $|G| = n$ and $|H| = d$, let
\[P_{H,G}(x_1, \ldots, x_n) \coloneqq \sum_{\Phi \in \mathrm{Hom}(H,G)} \prod_{v \in V(H)}x_{\Phi(v)} \in \mathcal{S}_d^n.\]
\end{definition}

The key fact is that $P_{H,G}$ can be computed by a small arithmetic circuit in the case when $H$ has small treewidth. For this we use the following lemma, proven in \cite{Brand2018Extensorcoding,fomin2012faster}.
\begin{lemma}\cite[Lemma 16]{Brand2018Extensorcoding}\label{fasteval}
Let $G$ and $H$ be graphs where $|G| = n$ and $|H| = d$. Then there is an arithmetic formula $C$ of size $O(d \cdot n^{\mathrm{tw}(H)+1})$ computing $P_{H,G}$. Furthermore, this formula can be constructed in time $O(1.76^d) + |C| \cdot \mathrm{polylog}(|C|)$.

\begin{proof}[Proof of  \Cref{twcount}]
We first construct a formula $C$ computing $P_{H,G}$ using \Cref{fasteval}. Note that $C$ can be evaluated on inputs in $\{\pm 1\}^n$ in time $O(n^{\mathrm{tw}(H)+1})$, and the maximum bit-complexity of $P_{H,G}$ on $\{\pm 1\}^n$ is $\log f(1,1,\ldots, 1) = \log (|\mathrm{Hom}(H,G)|) \le d \log n$.

Next note that $e_{n,d}(\partial \mathbf{x})P_{H,G}$ equals the number of injective homomorphisms from $H$ to $G$. Using \Cref{appl1} and the formula $C$ we first compute a $(1\pm \varepsilon)$ approximation to this number in time $4.075^d  n^{\mathrm{tw}(H)+O(1)}  \varepsilon^{-2} \log \varepsilon^{-1}$. In order to obtain a $(1 \pm \varepsilon)$ approximation to $\mathrm{Sub}(H,G)$ we divide this by $|\mathrm{Aut}(H,H)|$, which can be computed exactly in $O(1.01^d)$ time by using a $\mathrm{poly}(d)$-time reduction to graph isomorphism \cite{mathon1979note} and the quasi-polynomial time graph isomorphism algorithm of \cite{babai2016graph}.

The total time taken is \[O(1.76^d) + |C| \cdot \mathrm{polylog}(|C|) + 4.075^d \cdot n^{\mathrm{tw}(H)+O(1)} \cdot \varepsilon^{-2} \mathrm{polylog}(\varepsilon^{-1}) + O(1.01^d),\] 
\[\le 4.075^d \cdot n^{\mathrm{tw}(H)+O(1)} \cdot \varepsilon^{-2}\mathrm{polylog}(\varepsilon^{-1}).\qedhere\]
\end{proof}
\end{lemma}
\subsection{Lower Bounds on Perfectly Balanced Hash Families}

In this section we show how the bounds on $\mathbf{R}(e_{n,d})$ given in \cite{lee2016power} imply lower bounds on the size of perfectly balanced hash families.
\begin{definition}\cite[Definition 1]{alon2009balanced}
Let $n > \ell \ge k > 0$. A family of functions $\mathcal{F} = \{\pi : [n] \to [l]\}$ is said to be a perfectly-$k$ balanced hash family if for some $c \in \mathbb{N}$ and all $S \subseteq [n]$, the number of functions in $\mathcal{F}$ that are injective on $S$ equals $c$.
\end{definition}
\begin{theorem}\label{hashlower}
Let $\mathcal{F}$ be a perfectly-$k$ balanced hash family from $[n]$ to $[l]$. Then 
\begin{enumerate}[label=\alph*.]
\item If k is odd,
\[|\mathcal{F}| \ge \frac{\sum_{i=0}^{\lfloor k/2 \rfloor} \binom{n}{i}}{\sum_{i=0}^{\lfloor k/2 \rfloor}\binom{l}{i}}.\]
\item If k is even, \[|\mathcal{F}| \ge \frac{\left ( \sum_{i=0}^{ k/2} \binom{n}{i} \right ) - \binom{n-1}{k/2}}{\sum_{i=0}^{k/2}\binom{l}{i}}.\]
\end{enumerate}
\end{theorem}
\begin{proof}
Suppose that $k$ is odd, and let $\mathcal{F}$ be a perfectly $k$ balanced hash family from $[n]$ to $[l]$. For each $\pi \in \mathcal{F}$ define the linear forms $L_{\pi(i)} \coloneqq \sum_{j \in \pi^{-1}(i)} x_j$. Consider the polynomial 
\[f \coloneqq \sum_{\pi \in F} e_{k,l}(L_{\pi,1}, \ldots, L_{\pi,l}).\]
Since $\mathcal{F}$ is a perfectly balanced hash family it follows that, up to scaling, $f = e_{n,k}$, and hence $\mathbf{R}(f) = \sum_{i=0}^{\lfloor k/2 \rfloor} \binom{n}{i}$. On the other hand, by subadditivity of rank, we have that $\mathbf{R}(f) \le |\mathcal{F}| \mathbf{R}(e_{k,l}) = |\mathcal{F}|\sum_{i=0}^{\lfloor k/2 \rfloor} \binom{l}{i}$. Hence

\[|\mathcal{F}| \ge \frac{\sum_{i=0}^{\lfloor k/2 \rfloor} \binom{n}{i}}{\sum_{i=0}^{\lfloor k/2 \rfloor} \binom{l}{i}}.\]
The case for $k$ even is shown similarly.
\end{proof}
\section{Open Problems}\label{sec5}

\begin{question}
For all integers $u,v$ such that $u+v=d$, what is the minimum rank of a matrix with rows indexed by subsets of $[n]$ of size $u$ and columns indexed by subsets of $[n]$ of size $v$, such that entry $(I,J)$ is nonzero if and only if $I \cap J = \emptyset$, and entry $(I,J)$ equals entry $(K,L)$ whenever $I \cup J = K \cup L$? It follows from the method of partial derivatives that this quantity is a lower bound on $A(n,d)$. \Cref{catlower} shows that this is at most $2.6^d$.
\end{question}

\begin{question}
How many points are there in $\mathbb{C}^n$ such that the spaces spanned by any $d-1$ of them are contained in $\mathbf{V}(e_{n,d})$, but the spaces spanned by any $d$ of them are not? It is easy to see that $\mathbf{V}(e_{3,2})$ contains infinitely many such points; could it be that for all $d$ and some fixed $c \in \mathbb{N}$, $\mathbf{V}(e_{d+c,d})$ contains infinitely many such points? This would imply that $A(\mathbb{N},d) \le 2^d \mathrm{poly}(d)$. There is an obvious set of $n$ such points that was implicit in the upper bound of \Cref{pracub}, namely the standard basis vectors.
\end{question}
\begin{question}
Similarly, how many matrices in $\mathbb{C}^{n \times n}$ have the property that the span of any $d-1$ of them is contained in $\mathbf{V}(\mathrm{per}_d)$, but not the span of any $d$ of them? If there exist infinitely many points then it follows from \Cref{geomchar} and the fact that $\mathbf{R}(\mathrm{per}_d) \le 4^{d-1}$ \cite{landsberg2012tensors} that $A(\mathbb{N},d) \le 4^{d-1}$.
\end{question}

\begin{question}
Do all $(g,\varepsilon)$-support intersection certification algorithms require $\mathbf{R}_{\mathrm{supp}}(g)$ queries? \Cref{montest} shows that this is the case for monomials. Similarly, are $\mathbf{R}_\mathrm{supp}^\varepsilon(g)$ queries required to compute a $(1 \pm \varepsilon)$ approximation of $f(\partial \mathbf{x})g$ in the general black-box setting? \Cref{opteval} shows that this is true when $\varepsilon = 0$.
\end{question}

\begin{remark}
\Cref{masterthm} can be made algorithmic by using an explicit construction of a perfect splitter. The only such constructions we know however are far from optimal; that is, they give families of functions much larger than $\sigma(n,d,n_0,d_0)$ in general.
\end{remark}

\section{Acknowledgments}
I am very grateful to Ryan O'Donnell for numerous comments and suggestions, as well as feedback on an earlier draft of this paper. In particular, I would like to thank him for the proof of \Cref{detgen}. I would also like to thank Ryan Williams for comments on an earlier draft.
\bibliographystyle{alpha}
\bibliography{refs} 
\end{document}